\documentclass[aps,pra,twocolumn,floatfix,groupedaddress,superscriptaddress,footinbib,notitlepage,showkeys,showpacs,nofootinbib]{revtex4-1}
\usepackage{graphicx,graphics,times,bm,bbm,bbold,amssymb,amsmath,amsfonts,dsfont,color,xcolor,dcolumn}
\usepackage[bookmarksnumbered, colorlinks,plainpages]{hyperref}
\usepackage{diagbox}
\usepackage{multirow}
\usepackage{mathtools}
\usepackage{graphicx}
\usepackage{caption}
\usepackage{subcaption}
\usepackage{enumitem}
\hypersetup{colorlinks,linkcolor=blue,citecolor=blue,urlcolor=blue}
\usepackage{amsthm}

\newcommand{\ket}[1]{|#1\rangle}
\newcommand{\bra}[1]{\langle #1|}
\newcommand{\braket}[2]{\langle#1|#2\rangle}
\newcommand{\Id}{{\mathbb I}}
\newcommand{\Tr}{{\textrm {Tr}}}

\newcommand{\e}{{\mathrm {e}}}
\newcommand{\diag}{{\mathrm {diag}}}

\theoremstyle{definition}

\newtheorem{theorem}{Theorem}
\newtheorem{lemma}[theorem]{Lemma}
\newtheorem{proposition}[theorem]{Proposition}

\newtheorem{remark}[theorem]{Remark}
\newtheorem{definition}[theorem]{Definition}
\begin{document}

\title{
Mutually equibiased bases
}
\

\author{Seyed Javad Akhtarshenas}
\email{akhtarshenas@um.ac.ir}
\thanks{Corresponding author}
\affiliation{Department of Physics, Faculty of Science, Ferdowsi University of Mashhad, Mashhad, Iran}
\author{Saman Karimi}
\affiliation{Department of Physics, Faculty of Science, Ferdowsi University of Mashhad, Mashhad, Iran}
\author{Mahdi Salehi}
\affiliation{Department of Physics, Faculty of Science, Ferdowsi University of Mashhad, Mashhad, Iran}

\begin{abstract}
In the framework of mutually unbiased bases (MUBs),  a measurement in one basis gives  \emph{no information} about the outcomes of measurements in another basis. Here, we relax the  no-information condition by allowing the $d$ outcomes to be  predicted according to a predefined probability distribution $q=(q_0,\ldots,q_{d-1})$. The  notion of  mutual unbiasedness, however, is preserved by requiring that the extracted information is  the same for any preparation and any measurement; regardless of which state from which basis is chosen to prepare the system, the outcomes of measuring the system with respect to the other basis generate the same probability distribution. In  light of this, we define the notion of \emph{mutually equibiased bases} (MEBs) such that  within each basis the states are equibiased with respect to the states of the other basis and that the bases are mutually equibiased with respect to each other. For $d=2,3$, we derive a  set of $d+1$ MEBs. The mutual equibiasedness imposes  nontrivial constraints on the  distribution $q$, leading for $d=3$ to the restriction  $1/3\le\mu \le 1/2$ where $\mu=\sum_{k=0}^{2}q_k^2$.
To capture the incompatibility of the measurements in MEBs, we derive an inequality for the probabilities of projective measurements in a qudit system, which yields an associated entropic uncertainty inequality.  Finally, we  construct a  class of positive maps and their associated  entanglement witnesses based on MEBs.
While an entanglement witness constructed from MUBs is generally finer than one based on MEBs when both use the same number of bases, for certain values of the index $\mu$, employing a larger set of MEBs can yield a finer witness. We illustrate this behavior using isotropic states of a $3\times 3$ system. Our results reveal that not all bases in a set of $L$ MEBs can contribute to the entanglement detection. A  constraint, dependent on the probability distribution $q$, limits the number of MEBs that can be  employed.

\end{abstract}

\maketitle

\section{Introduction}\label{Sec-Introduction}
Quantum mechanics is full of mysteries, and perhaps the most profound is the existence of incompatible measurements, which imposes fundamental limitations on the precision of measurement outcomes through uncertainty and complementarity \cite{HeisenbergZfP1927,BohrNature1928}. This restricts the amount of information that can be extracted from measurements on a system, and it makes sense, therefore, that the more incompatible two measurements are, the greater the restrictions on the precision with which certain properties of a system can be simultaneously known.  The uncertainty relation quantifies how this incompatibility restricts the precision with which we can simultaneously measure incompatible observables.  Related to this, mutually unbiased bases (MUBs) play a crucial role, as measurements in one basis provide \emph{no information} about the outcomes of measurements in another basis \cite{SchwingerPNAS1960,WoottersAP1989,IvanovicJPA1981}.

In the $d$-dimensional Hilbert space $\mathcal{H}$, the orthonormal  bases $\{\ket{e_k^{(\alpha)}}\}_{k=0}^{d-1}$, numbered by $\alpha=1,\ldots, L$, are  said to be mutually unbiased if and only if the transition probability from any state of one set to any state of the second set is equal to $1/d$, irrespective of the chosen state, i.e.,  $|\braket{e_k^{(\alpha)}}{e_l^{(\beta)}}|^2=1/d$ for $\alpha\ne \beta$. For a comprehensive review of MUBs and their applications, see   \cite{ZyczkowskiIJQI2010} and for a review of MUBs in composite systems, see \cite{McNultyArxiv2024}.  In view of this, the set of projectors  $P_k^{(\alpha)}=\ket{e_k^{(\alpha)}}\bra{e_k^{(\alpha)}}$ for $k=0,\ldots,d-1$, forms a projective measurement on the $\alpha$th basis and the mutually unbiasedness implies  that
\begin{equation}\label{MUB1}
\Tr[P_k^{(\alpha)}P_l^{(\beta)}]=1/d+\delta_{\alpha\beta}(\delta_{kl}-1/d),
\end{equation}
for $k,l=0,\ldots,d-1$ and $\alpha,\beta=1,\ldots,L$.  Accordingly, if a system is prepared in a state belonging to one of the bases, then the outcomes of the measuring the system with respect to the other basis give, by no means, any information about the state of the system, i.e., all outcomes are predicted to occur with equal probability.

For a $d$-dimensional system, the number of MUBs is at most $d+1$, accordingly  $L\le d+1$. However, except for the particular cases where the dimension $d$ is a prime or a power of a prime number, for which there exist a complete set of $d+1$ MUBs, the maximum number of MUBs in other dimensions is not known.  For a general dimension $d$, it is known that at least three MUBs always exist.  It is shown that a  complete set of MUBs, if exists, forms a minimal and optimal set of projective measurements for quantum state determination \cite{WoottersAP1989,IvanovicJPA1981}. The mutual unbiasedness implies that the statistical errors are minimized when finite samples are measured \cite{WoottersAP1989,ZyczkowskiIJQI2010,RenesJMP2004}.

By generalizing the notion of unbiasedness from projective measurements to more general positive operator-valued measures,  Kalev and Gour \cite{KalevNJP2014} introduced the  notion of mutually unbiased measurements, which  encompass  a broader class of measurements that retain the essential features of unbiasedness.
Mutually unbiased bases and measurements have found important applications  in quantum information,  serving both as effective methods for entanglement detection \cite{SpenglerPRA2012,ChruscinskiPRA2014,ChenPRA2014,ChruscinskiPRA2018,LiIJTP2019,SalehiQIP2021},  and as foundational frameworks for entropic certainty and uncertainty relations \cite{MaassenPRL1988,ColesRMP2017,BallesterPRA2007,MolmerPRA2009,WehnerNJP2010,SanchezPLA1993,ZyczkowskiPRA2015,HuangPRA2024,HuangPRResearch2024}.
Whereas entropic uncertainty relations set lower bounds on average entropies, reflecting the impossibility of a state being a simultaneous eigenstate of all observables, entropic certainty relations set upper bounds, arising from the fact that, apart from the completely mixed state, no quantum state can produce $d$ equally likely outcomes for every observable at once \cite{SanchezPLA1993}.
For the average entropy characterizing measurements of a pure quantum state of dimension $d$ in $L$ orthogonal bases, both lower and upper bounds have been established in  \cite{ZyczkowskiPRA2015}.

Here we extend this notion of MUBs to   define mutually equibiased bases (MEBs) such that the standard MUBs are obtained as a special limiting case of our results. To this aim, the unbiasedness condition is partially relaxed  for MEBs in  the sense that the transition probability from any state of one basis to any state of the second basis is no longer necessarily equal to $1/d$.     In the most general setting of quantum measurement, if a system is prepared in a state belonging to one basis, the statistics of the outcomes of measuring the system with respect to the states of the other basis provides a   probability distribution  ${q}=(q_0,\ldots,q_{d-1})$.   We adapt this notion of biasedness; however, our restriction is that regardless of which state of the first basis is chosen, the probability distribution of the transition to  each state of the second basis remains the same up to a permutation of its entries. Accordingly,  within each basis the states are equibiased with respect to the states of the other basis and that the bases are mutually equibiased with respect to each other. In  light of this,  we associate a probability distribution ${q}=(q_0,\ldots,q_{d-1})$ with each of  MEBs such that
\begin{equation}\label{MEB1}
\Tr[P_k^{(\alpha)}P_l^{(\beta)}]=q_{\sigma^(k)\oplus \sigma(l)}+\delta_{\alpha\beta}(\delta_{kl}-q_{\sigma(k)\oplus \sigma(l)}),
\end{equation}
which reduces to Eq. \eqref{MUB1} of MUBs when $q_k=1/d$ for $k=0,\ldots,d-1$. Here, $\oplus$ stands for the sum module $d$, and  $\sigma(k)$ denotes the permutation of $k\in (0,1,\ldots,d-1)$.

In MEBs, in contrast to  MUBs,  the knowledge of  measuring the system with respect to the one basis, when the system is prepared in any state of the other basis,  leaks some information about the state of the system.  The mutual equibiasedness of the bases implies, however,  that they are biased in the same degree  in the sense that regardless of which state from which basis is chosen to prepare the system, the extracted  information is the same, i.e., they are predictable in the same degree with respect to each other.
It therefore makes sense  that the question of how much information can be extracted should be related to how far the probability distribution is from a homogenous distribution.
The index of coincidence, defined by $\mu(q)=\sum_{k=0}^{d-1}q_{k}^2$ for the probability  $q=(q_0,\ldots,q_{d-1})$, is a measure of how much a probability distribution deviates from a uniform distribution; as such, it can play a key role when the extracted information is quantified. Hereafter, for brevity, we will refer to the index of coincidence simply as the index  and denote it by $\mu$ for a $d$-outcome probability distribution $q=(q_0,\ldots,q_{d-1})$.

Related to the index $\mu$, we  now pose the  question of whether, for a given $d$, a  set of $d+1$  MEBs exists for all values of  $\mu\in[1/d, 1]$.
We derive a  set of $d+1$  MEBs for two cases $d=2,3$.
Our results show that except for the particular case of two-dimensional system for which a  set of three  MEBs exists for all $\mu\in[1/2,1]$, it is not the case for  $d\ge 3$.  More precisely,  for  $d=3$, the constraints imposed by the mutual equibiasedness limits the range of the index $\mu$ to the feasible region  $1/3\le \mu \le 1/2$ within which a set of four MEBs exists.

Quantum incompatibility  imposes limitation on  the information that can be extracted from measurements of different observables, and the entropic uncertainty relation provides a bound on how much information can be gained  from  such measurements.
To capture this notion of incompatibility  in MEBs, we derive an inequality for the probabilities of projective measurements in a qudit system, from which we obtain  an  entropic uncertainty inequality.

Among the various applications of MUBs, their role in entanglement detection is particularly important. Similarly, MEBs can also be used for identifying entanglement. The most general tool to detect entanglement  is the so-called entanglement witness, and a powerful tool  to construct entanglement witness is to use positive but not completely positive maps. Based on MEBs, we  construct a  class of positive maps and their associated  entanglement witnesses.

The paper is organized as follows. In Sec. \ref{Sec-MEBs} we briefly review the relevant concepts of bistochastic and unistochastic matrices, along with a parameterization for probability distribution ${q}=(q_0,\ldots,q_{d-1})$.  We then  introduce the notion of MEBs and  present a complete derivation of MEBs for the two simplest cases, $d=2,3$. In Sec. \ref{Sec-Entropic}, we derive   an inequality for the probabilities of projective measurements in MEBs of a qudit system. Based on this, we provide  an  entropic uncertainty inequality in MEBs. In Sec.  \ref{Sec-EntDetection}, we construct  a positive map based on MEBs and  then introduce  an entanglement witness and examine its performance  for two-qutrit isotropic states. We conclude the paper in Sec. \ref{Sec-Conclusion} with a brief summary. Some proofs of the results  are provided in Appendices \ref{Appendix-ProofProposition} and   \ref{Appendix-Eigenvalues-G}.

\section{Mutually equibiased bases (MEB\lowercase{s})}\label{Sec-MEBs}
\subsection{Preliminaries}\label{SubSection-Preliminaries}
\subsubsection{Bistochastic and unistochastic matrices}\label{SubsubSection-Bistochastic}
We begin this section by defining the notion of mutually equibiased bases, which generalizes the concept of mutually unbiased bases to encompass a broader family of  probability distributions. We will see that the allowed region of this family is shaped   by the mutually equibiased condition.
\begin{definition}
Let  $q=(q_0,\ldots,q_{d-1})$ be a probability  distribution over  $d$ outcomes. We say that the  set of $L$ orthonormal bases $\{\mathcal{B}^{(1)}, \ldots,\mathcal{B}^{(L)}\}$  of the $d$-dimensional Hilbert space  $\mathcal{H}$  is MEBs associated with    the probability  distribution $q$  if  for every pair $\alpha\ne \beta$  the corresponding  projection operators  satisfy Eq. \eqref{MEB1} for some permutation $\sigma$. Here   $\mathcal{B}^{(\alpha)}=\{\ket{e_k^{(\alpha)}}\}$ where $\alpha=1,\ldots,L$ and $k=0,\ldots,d-1$.
\end{definition}

The definition of MEBs relies on a probability distribution $q=(q_0,\ldots,q_{d-1})$, from which one can define the $d\times d$ matrix $\mathcal{Q}$ with entries $\mathcal{Q}_{kl}=q_{k\oplus l}$. Clearly, $\mathcal{Q}$ is bistochastic; its  entries  are nonnegative with each of its rows and columns summing to unity, i.e.,
\begin{eqnarray}
\mathcal{Q}_{kl}\ge 0, \qquad \sum_{k=0}^{d-1}\mathcal{Q}_{kl}=1, \qquad \sum_{l=0}^{d-1}\mathcal{Q}_{kl}=1.
\end{eqnarray}
More precisely,  the set of matrices $\mathcal{Q}$ forms a  $(d-1)$-dimensional subclass of the bistochastic matrices of size $d\times d$, which themselves comprise a convex set of dimension $(d-1)^2$, known as  Birkhoff's polytope \cite{ZyczkowskiCMP2005}.
A bistochastic matrix $\mathcal{Q}$ is called unistochastic matrix if and only if its entries are the squares of the absolute values of the entries of some unitary matrix.
In addition to the normalization requirement, i.e., in each row and column of a unitary matrix the squared moduli sum to unity, the moduli of its entries must also satisfy certain constraints arising from the mutual orthogonality of the rows and columns.  Accordingly, not all bistochastic matrices arise from unitary matrices; nonetheless,  the $d\times d$ bistochastic matrix $\mathcal{Q}$,  associated with the transition probabilities between two orthonormal bases  $\mathcal{B}^{(\alpha)}$ and $\mathcal{B}^{(\beta)}$   is a unistochastic matrix. The  characterization of  the unistochastic subset of Birkhoff's polytope, i.e., that for which there exist a unitary matrix $\mathcal{U}$ such that    $\mathcal{Q}_{kl}=|\mathcal{U}_{kl}|^2$, is considered in \cite{ZyczkowskiCMP2005}.

To set the convention, suppose $\mathcal{B}^{(1)}=\{\ket{e_k^{(1)}}\}_{k=0}^{d-1}$ denotes the standard (computational) basis of the $d$-dimensional Hilbert space $\mathcal{H}$.
Then, for each    $\beta\in\{2,\ldots,L\}$, the orthonormal basis $\mathcal{B}^{(\beta)}$ is related to $\mathcal{B}^{(1)}$ via a unitary matrix $\mathcal{U}^{(\beta)}$ as  $\ket{e_k^{(\beta)}}=\mathcal{U}^{(\beta)}\ket{e_k^{(1)}}$, for $k=0,\ldots,d-1$, where   $\mathcal{U}^{(\beta)}_{kl}=\braket{e_k^{(1)}}{e_l^{(\beta)}}$.
By equibiasedness, we have to assign    $|\mathcal{U}_{kl}^{(\beta)}|^2=\mathcal{Q}_{kl}$, for all $\beta\in\{2,\ldots,L\}$,    where the $d\times d$ bistochastic matrix $\mathcal{Q}$ is defined by  $\mathcal{Q}_{kl}=q_{k\oplus l}$ for the  probability distribution $q=(q_0,\ldots,q_{d-1})$.
An important question is to determine for which probability distributions $q$  the associated  bistochastic matrix $\mathcal{Q}$ is unistochastic, that is,  which  $q$ gives rise to unistochastic matrix  $\mathcal{Q}$ corresponding to  unitary matrices $\mathcal{U}^{(\beta)}$ such that $\mathcal{Q}_{kl}=|\mathcal{U}_{kl}^{(\beta)}|^2$.

Before we proceed further, let us first note that for a given unistochastic matrix $\mathcal{Q}$, if one unitary matrix $\mathcal{U}$ is found such that $\mathcal{Q}_{kl}=|\mathcal{U}_{kl}|^2$, then there exist infinitely many unitary matrices   that satisfy this condition.
This follows from the fact that the left and right multiplication of  a unitary matrix $\mathcal{U}$ by any diagonal unitary matrices  $\mathcal{D}_1$ and $\mathcal{D}_2$ results in another unitary matrix $\mathcal{V}=\mathcal{D}_1\mathcal{U}\mathcal{D}_2$, but the associated unistochastic matrix remains unchanged.
This defines an equivalence relation as $\mathcal{U}\sim\mathcal{V}=\mathcal{D}_1\mathcal{U}\mathcal{D}_2$ \cite{ZyczkowskiJMP2009}.
Note, however, that the  equivalence is partial in general;  not any pair of unitary matrices $\mathcal{U}$ and $\mathcal{V}$ generating the same unistochastic matrix $\mathcal{Q}$ can be written as $\mathcal{V}=\mathcal{D}_1 \mathcal{U}\mathcal{D}_2$ for some diagonal unitary matrices $\mathcal{D}_1$ and $\mathcal{D}_2$.  For  $d=2$, the correspondence between a unitary matrix and its
unistochastic image is one to one up to the multiplications from the  left and right  by diagonal unitary matrices.
For higher dimensions,  the equivalence is not complete, meaning that  inequivalent unitaries   can generate the same unistochastic matrix. An  example occurs in
dimension $d=4$ for a one-parameter class of complex Hadamard matrices  \cite{ZyczkowskiOSID2006}. For $d=3$,  to the best of our knowledge,   there is no general proof that inequivalent unitaries cannot exist, nor any counterexample indicating that such inequivalent unitaries do exist.

What is crucial in our analysis, however,  is the requirement of the \emph{mutual} equibiasedness. This condition imposes a severe restriction on the set of MEBs, making it finite and constrained.
To be more precise,  suppose  $\mathcal{Q}$ is a unistochastic matrix generated by  $L-1$ unitary matrices  $\mathcal{U}^{(\beta)}$ such that $\mathcal{Q}_{kl}=|\mathcal{U}_{kl}^{(\beta)}|^2$ for each    $\beta\in\{2,\ldots,L\}$.  The mutually equibiasedness requires further that  they are mutually unistochastic, i.e.,  $|[{\mathcal{U}^{(\alpha)}}^\dagger\mathcal{U}^{(\beta)}]_{kl}|^2=\mathcal{Q}_{\sigma(k)\sigma(l)}$ for all $\alpha,\beta\in\{2,\ldots,L\}$ with $\alpha\ne \beta$.
The permutation $\sigma$, appearing here and in Eq. \eqref{MEB1}, depends on $\alpha$ and $\beta$ in general, however, when one of the bases corresponds to the standard basis, we follow the convention that the permutation $\sigma$ does nothing, i.e., $\sigma(k)=k$ and $\sigma(l)=l$.

\subsubsection{Parametrization of  the distribution $q=(q_0,\ldots,q_{d-1})$}\label{Section-parameterization}
For an arbitrary  $d$-outcome probability distribution $q=(q_0,\ldots,q_{d-1})$, we define  a $(d-1)-$dimensional Hermitian matrix $N$ such that it has nonzero  entries  only  along the antidiagonal, i.e.,  $N_{kk^\prime}\ne 0$ if and only if  $k\oplus k^\prime=d$. More precisely, let   $N_{kk^\prime}=\sum_{m=0}^{d-1}\omega^{k^\prime m}q_{m} =\sum_{m=0}^{d-1}\omega^{-k m}q_{m}$ so that for nonzero entries we have
\begin{eqnarray}\label{N-k-d-k}
N_{k,d-k}=\sum_{m=0}^{d-1}\omega^{-k m}q_{m},
\end{eqnarray}
for $k=1,\ldots,d-1$ where $\omega=\e^{2\pi i/d}$.

The matrix $N$ captures the non-uniformity of the probability distribution $q=(q_0,\ldots,q_{d-1})$ in the sense that $N=0$ if and only if the underlying distribution is uniform, i.e., $q=(1/d,\ldots,1/d)$. Clearly,   $N$  has $\lceil \frac{d-1}{2} \rceil$  independent   entries, where  $\lceil \frac{d-1}{2} \rceil$ denotes the smallest integer greater than or equal to $\frac{d-1}{2}$, that is, it equals $\frac{d-1}{2}$ if $d$ is odd and $\frac{d}{2}$ if $d$ is even. All entries are complex, except for   $N_{d/2,d/2}$ which is real  and occurs  when $d$ is even.   With this notation,  the $(d-1)-$dimensional  matrix  $N$ is fully determined  by $d-1$ real independent parameters. Accordingly, by writing $N_{k,d-k}=\delta_{2k-1}+i\delta_{2k}$, we have for  the $(d-1)$ real  parameters $\delta_{1},\ldots,\delta_{d-1}$
\begin{eqnarray}\label{delta-odd}
\delta_{2k-1}&=&\sum_{m=0}^{d-1}\cos\bigg[\frac{2\pi k m}{d}\bigg]q_{m}, \\ \label{delta-even}
\delta_{2k}&=& -\sum_{m=0}^{d-1}\sin\bigg[\frac{2\pi k m}{d}\bigg]q_{m},
\end{eqnarray}
for $k=1,\ldots,\lceil \frac{d-1}{2} \rceil$. This parameterization is particularly useful to quantify the extent to which the probability distribution $q=(q_0,\ldots,q_{d-1})$ deviates from the uniform distribution $q=(1/d,\ldots,1/d)$.
To illustrate  the nature of this parameterization, we  explicitly consider two cases of $d=2$ and $d=3$.

For $d=2$,   the only nonzero entry of $N$ is given by $N_{11}=\delta$, where $\delta=q_{0}-q_{1}$  such that $-1\le \delta\le 1$.
It follows therefore that
\begin{eqnarray}\label{delta-q0q1}
q_0=(1+\delta)/2, \quad q_1=(1-\delta)/2.
\end{eqnarray}

For $d=3$, the nonzero entries of $N$ are given by $N_{12}=N_{21}^\ast=\delta_1+i\delta_2$. Two independent parameters are given by
\begin{eqnarray}\label{delta1delta2}
\delta_1=(3q_0-1)/2, \quad \delta_2=\sqrt{3}(q_{2}-q_{1})/2,
\end{eqnarray}
where
\begin{eqnarray}
-1/2\le &\delta_1&\le 1, \\
(-1+\delta_1)/\sqrt{3}\le &\delta_2&\le (1-\delta_1)/\sqrt{3}.
\end{eqnarray}
We therefore have
\begin{eqnarray}
q_0&=&1/3+2\delta_1/3, \\
q_1&=&1/3-\delta_1/3-\delta_2/\sqrt{3}, \\
q_2&=&1/3-\delta_1/3+\delta_2/\sqrt{3}.
\end{eqnarray}

In the following subsections, we examine two special cases of $d=2$ and $d=3$, providing explicite construction of  MEBs in each case.

\subsection{Mutually equibiased bases for $d=2$}
For the  simplest case of a  two-level system, suppose $\mathcal{B}^{(1)}=\{\ket{0},\ket{1}\}$ represents the standard basis. Associated with the probability distribution ${q}=(q_0,q_1)$, or  its corresponding bistochastic matrix  $\mathcal{Q}=\begin{psmallmatrix}q_0 & q_1 \\ q_1 & q_0 \end{psmallmatrix}$,  we have to construct the unitary matrix $\mathcal{U}$ with prescribed moduli of its matrix elements. Starting with $\mathcal{U}=\begin{psmallmatrix}\sqrt{q_0} & \sqrt{q_1} \\ \sqrt{q_1} & \sqrt{q_0}\e^{i\xi} \end{psmallmatrix}$,   unitarity requires $\xi=\pi$, hence
\begin{equation}
\mathcal{U}=\begin{pmatrix}
  \sqrt{q_0} & \sqrt{q_1}  \\
  \sqrt{q_1} & -\sqrt{q_0} \\
  \end{pmatrix}.
\end{equation}
Given that  for any diagonal unitary matrices  $\mathcal{D}_1$ and $\mathcal{D}_2$, the unistochastic matrix associated with  $\mathcal{V}=\mathcal{D}_1\mathcal{U}\mathcal{D}_2$ is  the same as that associated with  $\mathcal{U}$ \cite{ZyczkowskiJMP2009}, one can construct  a second unitary matrix $\mathcal{V}$ from $\mathcal{U}$ as $\mathcal{V}=D\mathcal{U}$ where $D=\diag\{1,\e^{i\alpha}\}$. We obtain
\begin{equation}
\mathcal{V}=\begin{pmatrix}
  \sqrt{q_0} & \sqrt{q_1}  \\
  \sqrt{q_1}\e^{i\alpha} & -\sqrt{q_0}\e^{i\alpha} \\
  \end{pmatrix}.
\end{equation}
Although both $\mathcal{U}$ and $\mathcal{V}$ generate the same bistochastic matrix $\mathcal{Q}$  and both define bases that are MEBs with respect to the standard basis, they are not necessarily MEBs with respect to each other. For this to hold, $\mathcal{U}$ and $\mathcal{V}$ must be mutually equibiased, i.e., the unitary matrix $\mathcal{V}^\dagger\mathcal{U}$  must generate the same bistochastic  matrix $\mathcal{Q}$. This places a constraint on  $\alpha$ such that $q_0^2+q_1^2+2q_0q_1\cos{\alpha}=q_0$, which uniquely determines $\alpha$ as
\begin{equation}
\alpha=\cos^{-1}[\delta/(1+\delta)],
\end{equation}
with   the unbiasedness parameter $\delta=q_0-q_1$  defined by Eq. \eqref{delta-q0q1}.

Starting from the standard orthonormal basis $\{\ket{e_k^{(1)}}=\ket{k}\}$, one can apply the unitary matrices $\mathcal{U}$ and $\mathcal{V}$ to construct new bases $\{\ket{e_k^{(2)}}=\mathcal{U}\ket{e_k^{(1)}}\}$ and  $\{\ket{e_k^{(3)}}=\mathcal{V}\ket{e_k^{(1)}}\}$      for  $k=0,1$. This notion of equibiasedness admits an interesting geometric interpretation in the Bloch vector representation, in the sense that  with  $P_k^{(\alpha)}=(\Id_2+\boldsymbol{r}_k^{(\alpha)}\cdot \boldsymbol{\sigma})/2$ as the Bloch vector representation of the projection operator $P_k^{(\alpha)}$, the Bloch vectors satisfy
\begin{eqnarray}
\boldsymbol{r}_k^{(\alpha)}\cdot \boldsymbol{r}_l^{(\beta)}=
 (-1)^{k\oplus l}\;\delta,  & \mbox{if }\; \alpha\ne\beta.
\end{eqnarray}
Accordingly, the Bloch vectors form a symmetric or equiangular triple such that they are mutually symmetric with the same mutual angle  $\vartheta=\cos^{-1}[\delta]$ (see Fig. \ref{FIG-Triangle2}).
\begin{figure}[t]
\includegraphics[scale=0.4]{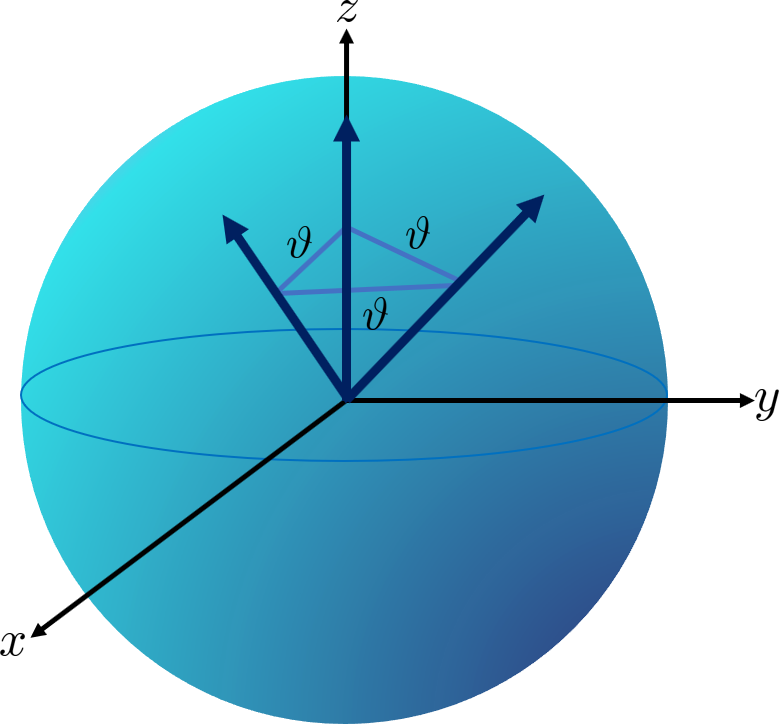}
\caption{Bloch vector representation for MEB, when $d=2$. Geometrically, equibiasedness means equal angle, $\vartheta=\cos^{-1}[\delta]$ between Bloch vectors of different bases.  }
\label{FIG-Triangle2}
\end{figure}

\subsection{Mutually equibiased bases for $d=3$}
In this case, associated with the probability distribution $q=(q_0,q_1,q_{2})$, we have to construct the unitary matrix $\mathcal{U}$ with prescribed moduli of its matrix elements as
\begin{equation}\label{U3=XiEta}
\mathcal{U}=\begin{pmatrix}
  \sqrt{q_0} & \sqrt{q_1} & \sqrt{q_2} \\
  \sqrt{q_1} & \sqrt{q_2}\e^{i\xi_{0}} & \sqrt{q_0}\e^{i\xi_{1}} \\
  \sqrt{q_2} & \sqrt{q_0}\e^{i\eta_{1}} & \sqrt{q_1}\e^{i\eta_{0}}
\end{pmatrix}.
\end{equation}
The unitarity condition $\mathcal{U}\mathcal{U}^\dagger=\Id$, i.e., the mutual orthogonality of the rows (and columns), gives simply the  conditions
\begin{eqnarray}\label{Tri1}
& &L_2+L_0\e^{i\xi_{0}}+L_1\e^{i\xi_{1}}=0, \\
& &L_1+L_2\e^{i\eta_{1}}+L_0\e^{i\eta_{0}}=0, \\ \label{Tri3}
& &L_0+L_1\e^{i(\xi_0-\eta_{1})}+L_2\e^{i(\xi_1-\eta_{0})}=0,
\end{eqnarray}
where we have defined \cite{ZyczkowskiJMP2009}
\begin{equation}\label{L0L1L2}
L_0=\sqrt{q_1q_2}, \quad L_1=\sqrt{q_2q_0}, \quad L_2=\sqrt{q_0q_1}.
\end{equation}
Equations \eqref{Tri1}-\eqref{Tri3} imply that the unitarity is fulfilled if we find the phases in such a way that three  line segments $L_0$, $L_1$, and $L_2$ form a triangle. This imposes the triangle inequality on the line segments \cite{ZyczkowskiCMP2005}
\begin{equation}\label{L-inequalities}
|L_j-L_k|\le L_i\le L_j+L_k,
\end{equation}
where  $(i,j,k)$ are different choices of   $(0,1,2)$.   For  probabilities $q=(q_0,q_1,q_2)$ for which the  above inequalities hold, Eqs.  \eqref{Tri1}-\eqref{Tri3} give the phases as
\begin{eqnarray}\label{d3-xi0}
\cos{\xi_0}&=&\frac{L_1^2-L_0^2-L_2^2}{2L_0L_2}=\cos{(\eta_0-\eta_1)}, \\ \label{d3-xi1-eta1}
\cos{\xi_1}&=&\frac{L_0^2-L_1^2-L_2^2}{2L_1L_2}=\cos{\eta_1}, \\ \label{d3-eta0}
\cos{(\xi_0-\xi_1)}&=&\frac{L_2^2-L_0^2-L_1^2}{2L_0L_1}=\cos{\eta_0}.
\end{eqnarray}
For each probability distribution $q=(q_0,q_1,q_2)$ satisfying   the chain-link conditions \eqref{L-inequalities},  the  associated unitary matrix $\mathcal{U}$,  defined in Eq. \eqref{U3=XiEta} with phases specified in Eqs. \eqref{d3-xi0}-\eqref{d3-eta0},  generates the corresponding unistochastic matrix $\mathcal{Q}$. By construction, the matrix  $\mathcal{U}$ represents the most general form consistent with $q=(q_0,q_1,q_2)$ and defines the second orthonormal basis $\{\ket{e_k^{(2)}}=\mathcal{U}\ket{e_k^{(1)}}\}$ with respect to the standard basis $\{\ket{e_k^{(1)}}\}$.
Given   $\mathcal{U}$, we have  to construct two  unitary matrices, $\mathcal{V}$ and $\mathcal{W}$ such that all bases are mutually equibiased.  The following proposition provides a   set of four MEBs for $d=3$.

\begin{proposition}\label{Proposition-d3}
Starting from the standard basis $\{\ket{e_{k}^{(1)}}\}_{k=0}^{2}$, the unitary matrices $\mathcal{U}$, $\mathcal{V}$, and $\mathcal{W}$  generate a one-parameter class  of four MEBs. Here    $\mathcal{V}=\mathcal{D}_{\alpha}\mathcal{U}$  and $\mathcal{W}=\mathcal{D}_{\beta}\mathcal{U}$, where  $\mathcal{D}_{\alpha}=\diag\{1,\e^{i\alpha_1},\e^{i\alpha_2}\}$ and $\mathcal{D}_{\beta}=\diag\{1,\e^{i\beta_1},\e^{i\beta_2}\}$, respectively, and $\mathcal{U}$ is defined  by Eq. \eqref{U3=XiEta}.
The pairs $(\alpha_{1},\alpha_{2})$ and $(\beta_{1}, \beta_{2})$ are distinct and take their  values from the  set
\begin{eqnarray}\label{Pairs-alpha}
(0,\theta^{\mu}), \qquad (\theta^{\mu},0), \qquad (\theta^{\mu},\theta^{\mu}),
\end{eqnarray}
where
\begin{eqnarray}\label{theta-mu2}
\theta^{\mu}=\cos^{-1}{\Big[\frac{-\mu}{1-\mu}\Big]},
\end{eqnarray}
for  $\mu\in[1/3,1/2]$. Furthermore, the probabilities $(q_0,q_1,q_2)$ are different permutations of $(q^{\mu}_{+}, q^{\mu}_{0}, q^{\mu}_{-})$ where
\begin{eqnarray}\label{q0-mu}
q^{\mu}_0&=&\frac{1-\mu}{2}, \\ \label{qpm-mu}
q^\mu_{\pm}&=&\frac{1}{4}\Big[(1+\mu)\pm\sqrt{(1+\mu)^2-4(1-\mu)^2}\Big].
\end{eqnarray}
\end{proposition}

The proof of the proposition is given in Appendix \ref{Appendix-ProofProposition};  here we provide a brief description of the structure of the obtained MEBs.  In Fig. \ref{FIG-Triangle-Solution}  we illustrate the one-parameter class of solutions within  the feasible region.   The feasible region, shown in yellow,   is bounded by three hypocycloids each corresponding to the saturation of one of the chain-link inequalities \eqref{L-inequalities}, and three circular arcs of radius $\mu=1/2$. The  hypocycloids are obtained from the condition that the  matrix $\mathcal{U}$ in Eq. \eqref{U3=XiEta} is unitary, generating  the unistochastic matrix $\mathcal{Q}$ with entries $\mathcal{Q}_{kl}=q_{k\oplus l}$. The three circular arcs, on the other hand, follow from the condition that the unitary matrix $\mathcal{V}^\dagger\mathcal{U}$, with $\mathcal{V}=\mathcal{D}_{\alpha}\mathcal{U}$,  must generate the same unistochastic  matrix $\mathcal{Q}$, up to a permutation of its entries. The latter imposes further constraints on the solutions, restricting them to a one-parameter class of probability distributions, as stated in the proposition.
The   class consists of  six curves, each starting at  $\mu=1/3$ at the center of the feasible region and ending  at  $\mu=1/2$ on the boundary.  The curve colored in  darker blue, for example, corresponds to  $q=(q^\mu_{+}, q^\mu_0, q^\mu_{-})$ when   $\alpha_1=0$ and $\alpha_2=\cos^{-1}[-\mu/(1-\mu)]$, and all other curves can be obtained by any permutation of its probability components.

\begin{figure}[t]
\includegraphics[scale=0.9]{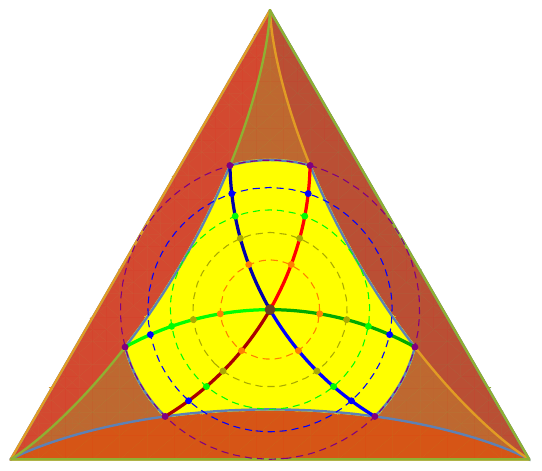}
\caption{Geometrical representation of the probability distribution $q=(q_0,q_1,q_2)$, for which  a  set of four MEBs exist. The figure is plotted in the $\delta_2\delta_1$ plane,  where  independent parameters $\delta_1$ and $\delta_2$ are defined in Eq. \eqref{delta1delta2}.  The curves represent the allowed values of the  distribution   $q=(q^\mu_{+}, q^\mu_0, q^\mu_{-})$, or any permutation of its probabilities, when   $\alpha_1=0$ and $\alpha_2=\cos^{-1}[-\mu/(1-\mu)]$.  The circles correspond to regions with constant index $\mu$. Corresponding to each $\mu$, only six points, i.e., the distribution $q$ and its permutations,  are allowed for a given $\mu$.  }
\label{FIG-Triangle-Solution}
\end{figure}

It is worth noting that, although all three unitary matrices $\mathcal{U}$, $\mathcal{V}$, and $\mathcal{W}$ generate the same unistochastic matrix $Q$, associated with  the probability distribution $q=(q_0,q_1,q_2)$, the unistochastic matrices resulting from their mutual products $\mathcal{V}^\dagger \mathcal{U}$, $\mathcal{W}^\dagger \mathcal{U}$, and $\mathcal{W}^\dagger \mathcal{V}$  are not necessary equal to $Q$. However, they are equivalent in the sense that they are generated either  by
$q=(q_0,q_1,q_2)$ or by a permutation of its components,    such as $q^{(012)}=(q_1,q_2,q_0)$ and $q^{(210)}=(q_2,q_1,q_0)$, where $(012)$ and $(210)$ are forward and backward cycles of the group of permutation of three elements $\{0,1,2\}$.   All such possibilities are summarized in Table \ref{Table-AlphaBeta1} in Appendix \ref{Appendix-ProofProposition}.

Note that, in Proposition \ref{Proposition-d3}, we assume  that any two unitary matrices $\mathcal{U}$ and $\mathcal{V}$ generating the same unistochastic matrix $\mathcal{Q}$ can be written as $\mathcal{V}=\mathcal{D}_1 \mathcal{U}\mathcal{D}_2$ for some diagonal unitary matrices $\mathcal{D}_1$ and $\mathcal{D}_2$. However,  for $d=3$,  the existence  of inequivalent unitary matrices cannot be excluded;  therefore, whether this one-parameter class is exhaustive remains an open question.

\subsubsection{Example: $d=3$ and    $\mu=1/2$}
For $\mu=1/2$ and by choosing the nonincreasing order for probabilities, i.e., $q_0=q_{+}^{\mu}$, $q_1=q_{0}^{\mu}$, and $q_2=q_{-}^{\mu}$, we find from Eqs. \eqref{q0-mu} and \eqref{qpm-mu}
\begin{eqnarray}
q_0=(3+\sqrt{5})/8,\quad q_1=1/4,\quad q_2=(3-\sqrt{5})/8.
\end{eqnarray}
Using these in Eqs \eqref{Tri1} and \eqref{d3-xi0}-\eqref{d3-eta0}, we  get $\eta_0=0$, $\eta_1=\xi_0=\xi_1=\pi$.
Equation \eqref{theta-mu2} gives $\theta^{\mu}=\pi$, which gives  for  case (a)
\begin{eqnarray}
\alpha_1=\beta_2=0,\qquad  \alpha_2=\beta_1=\pi.
\end{eqnarray}
Using these in Eq. \eqref{U3=XiEta}, we can find  $\mathcal{U}$, from which $\mathcal{V}$ and $\mathcal{W}$ can also be obtained as
\begin{equation}
\mathcal{U}=\begin{pmatrix}
  \sqrt{q_0} & \sqrt{q_1} & \sqrt{q_2} \\
  \sqrt{q_1} & -\sqrt{q_2} & -\sqrt{q_0} \\
  \sqrt{q_2} & -\sqrt{q_0} & \sqrt{q_1}
\end{pmatrix},
\end{equation}
\begin{equation}
\mathcal{V}=\begin{pmatrix}
  \sqrt{q_0} & \sqrt{q_1} & \sqrt{q_2} \\
  \sqrt{q_1} & -\sqrt{q_2} & -\sqrt{q_0} \\
  -\sqrt{q_2} & \sqrt{q_0} & -\sqrt{q_1}
\end{pmatrix},
\end{equation}
\begin{equation}
\mathcal{W}=\begin{pmatrix}
  \sqrt{q_0} & \sqrt{q_1} & \sqrt{q_2} \\
  -\sqrt{q_1} & \sqrt{q_2} & \sqrt{q_0} \\
  \sqrt{q_2} & -\sqrt{q_0} & \sqrt{q_1}
\end{pmatrix}.
\end{equation}
Obviously,  all three unitaries generate the same unistochastic matrix $\mathcal{Q}$. They are also mutually equibiased as
\begin{equation}
\mathcal{V}^\dagger\mathcal{U}=\begin{pmatrix}
  \sqrt{q_0} & \sqrt{q_1} & -\sqrt{q_2} \\
  \sqrt{q_1} & -\sqrt{q_2} & \sqrt{q_0} \\
  -\sqrt{q_2} & \sqrt{q_0} & \sqrt{q_1}
\end{pmatrix},
\end{equation}
\begin{equation}
\mathcal{W}^\dagger\mathcal{U}=\begin{pmatrix}
  \sqrt{q_1} & \sqrt{q_2} & \sqrt{q_0} \\
  \sqrt{q_2} & \sqrt{q_0} & -\sqrt{q_1} \\
  \sqrt{q_0} & -\sqrt{q_1} & -\sqrt{q_2}
\end{pmatrix},
\end{equation}
\begin{equation}
\mathcal{W}^\dagger\mathcal{V}=\begin{pmatrix}
  \sqrt{q_2} & \sqrt{q_0} & \sqrt{q_1} \\
  \sqrt{q_0} & -\sqrt{q_1} & \sqrt{q_2} \\
  \sqrt{q_1} & \sqrt{q_2} & -\sqrt{q_0}
\end{pmatrix},
\end{equation}
which generate three equivalent unistochastic matrices $\mathcal{Q}$,  $\mathcal{Q}^{(123)}$,  and $\mathcal{Q}^{(321)}$, respectively.

\section{Entropic uncertainty relation for MEB\lowercase{s}}\label{Sec-Entropic}
Before proceeding with the  main theorems  of this section,  let us first present the following  lemma.
Suppose a set of projectors  $P_j^{(\alpha)}=\ket{e_j^{(\alpha)}}\bra{e_j^{(\alpha)}}$, for $j=0,\ldots,d-1$, forms a projective measurement on the $\alpha$th basis of the $d$-dimensional system $\mathcal{H}$. Suppose also that, for a suitable choice of the probability distribution $q=(q_0,\ldots,q_{d-1})$ and for $\alpha=1,\ldots,L$, the bases are  mutually equibiased in the sense of Eq. \eqref{MEB1}.  With the anti-diagonal Hermitian matrix $N_{k,d-k}$, defined by Eq. \eqref{N-k-d-k}, we have
\begin{eqnarray}
|N_{k,d-k}|=\bigg[\mu+2\sum_{m<m^\prime}^{d-1}q_{m}q_{m^\prime}\cos\Big(\frac{2\pi k (m-m^\prime)}{d}\Big)\bigg]^{1/2},
\end{eqnarray}
where $\mu=\sum_{j=0}^{d-1}q_j^2$ is the index of coincidence  of the distribution $q$. Define $\lambda_{\max}(N)=\max_{k}|N_{k,d-k}|$, where the maximum is taken over $k\in\{1,\ldots,(d-1)/2\}$ when $d$ is odd,  and over  $k\in\{1,\ldots,d/2\}$ when $d$ is even.

Let us now state the following lemma, which is essential for  the remainder of this section.
\begin{lemma}\label{Lemma-Reciprocal}
For a  two-qudit composite system $\mathcal{H}\otimes \mathcal{H}$, consider the  set of vectors
\begin{eqnarray}\label{two-qudit-phiAlpha}
\ket{\phi_k^{(\alpha)}}&=&\frac{1}{\sqrt{d}}\sum_{j=0}^{d-1}\omega^{kj}\ket{e_j^{(\alpha)}}\ket{e_j^{(\alpha)}}^\ast,
\end{eqnarray}
where $\omega=\e^{2\pi i/d}$,  $k=1,\ldots, d-1$,  and $\alpha=1,\ldots,L$. Here   $\ket{e_j^{(\alpha)}}^\ast$ denotes the complex conjugate of $\ket{e_j^{(\alpha)}}$ with respect to a definite basis, say the standard basis $\{\ket{e_j^{(1)}}\}_{j=0}^{d-1}$.
Then
\begin{enumerate}[label=(\alph*)]
\item\label{Lem-a}
The set of vectors \eqref{two-qudit-phiAlpha} is linearly independent and span an $L(d-1)$-dimensional subspace of $\mathcal{H}\otimes \mathcal{H}$ if and only if
\begin{equation}\label{Constraint-lambdaN}
\lambda_{\max}(N)|<1/(L-1).
\end{equation}
\item\label{Lem-b}
When \ref{Lem-a} holds,   the reciprocal vector  $\ket{\widetilde{\phi}_k^{(\alpha)}}$, corresponding to each vector $\ket{\phi_k^{(\alpha)}}$,  is defined by
\begin{equation}\label{reciprocal}
\ket{\widetilde{\phi}_{k}^{(\alpha)}}=\sum_{k^\prime=1}^{d-1}\sum_{\alpha^\prime=1}^{L}[G^{-1}]_{k^\prime k}^{(\alpha^\prime \alpha)} \ket{\phi_{k^\prime}^{(\alpha^\prime)}},
\end{equation}
where  $G^{-1}$ is the inverse of the Gram matrix $G$, defined by $G_{kk^\prime}^{(\alpha\alpha^\prime)}=\braket{\phi_k^{(\alpha)}}{\phi_{k^\prime}^{(\alpha^\prime)}}$, and  reciprocity means $\braket{\phi_k^{(\alpha)}}{\widetilde{\phi}_{k^\prime}^{(\alpha^\prime)}}=\delta_{kk^\prime}\delta^{(\alpha\alpha^\prime)}$.
\end{enumerate}
\end{lemma}
\begin{proof}
To prove  part \ref{Lem-a},  consider the Gram matrix $G$, defined by
\begin{eqnarray}\label{GramNkkp}
G_{kk^\prime}^{(\alpha\alpha^\prime)}&=&\braket{\phi_k^{(\alpha)}}{\phi_{k^\prime}^{(\alpha^\prime)}} \\ \nonumber
&=&N_{kk^\prime}+\delta_{\alpha\alpha^\prime}(\delta_{kk^\prime}-N_{kk^\prime}),
\end{eqnarray}
where $N_{kk^\prime}$ are entries of the matrix $N$, as
\begin{equation}\label{N-k-d-k-2}
N_{kk^\prime}=\frac{1}{d}\sum_{j,j^\prime=0}^{d-1}\omega^{k^\prime j^\prime-kj}q_{j\oplus j^\prime},
\end{equation}
  for $k,k^\prime=1,\ldots,d-1$. Clearly, $N$ is anti-diagonal and  its nonzero entries  can be simplified as Eq. \eqref{N-k-d-k}.
In Appendix  \ref{Appendix-Eigenvalues-G} we derive the eigenvalues of $G$;  in particular, its smallest and largest eigenvalues are $\lambda_{\min}(G)=1-(L-1)\lambda_{\max}(N)$ and $\lambda_{\max}(G)=1+(L-1)\lambda_{\max}(N)$, respectively. It follows therefore that the set of vectors \eqref{two-qudit-phiAlpha} is linearly independent if and only if the Gram matrix is positive definite, i.e., $\lambda_{\min}(G)>0$, thereby completing  the proof of part \ref{Lem-a}.

To prove  part \ref{Lem-b} of the lemma, suppose vectors \eqref{two-qudit-phiAlpha} are  linearly independent, i.e., $\lambda_{\max}(N)<1/(L-1)$. The reciprocal vectors can be expanded in terms of the original vectors as
\begin{equation}\label{reciprocal}
\ket{\widetilde{\phi}_{k}^{(\alpha)}}=\sum_{k^\prime=1}^{d-1}\sum_{\alpha^\prime=1}^{L}C_{kk^\prime}^{(\alpha\alpha^\prime)} \ket{\phi_{k^\prime}^{(\alpha^\prime)}}.
\end{equation}
By multiplying  $\bra{\phi_l^{(\beta)}}$ from the left in this equation, we get
\begin{eqnarray}
\delta_{kl}\delta^{(\alpha\beta)}&=&\sum_{k^\prime=1}^{d-1}\sum_{\alpha^\prime=1}^{L}C_{kk^\prime}^{(\alpha\alpha^\prime)}G_{l k^\prime}^{(\beta\alpha^\prime)} \\ \nonumber
&=&\left[CG^\textrm{T}\right]_{kl}^{(\alpha\beta)}.
\end{eqnarray}
This implies that $CG^\textrm{T}=\Id_{L(d-1)}$, or equivalently $C^\textrm{T}=G^{-1}$, where ${\textrm{T}}$ stands for the transposition. Accordingly, up to a transposition,  the matrix of the coefficients is obtained by the inverse of the Gram matrix $G$. This completes the proof of the lemma.
\end{proof}

\subsection{Inequality for the probabilities of projective measurements in MEBs of a qudit system}
The following theorem   holds.
\begin{theorem}\label{Theorem-EUR1}
Suppose that, for a suitable choice of the probability distribution $q=(q_0,\ldots,q_{d-1})$, and for $\alpha=1,\ldots,L$,  there exist $L$  MEBs in a  $d$-dimensional Hilbert space $\mathcal{H}$.  Then, for an arbitrary $d$-dimensional quantum state $\varrho$, the probability of obtaining the $j$th result when the system is measured on the $\alpha$th basis is given by $p_{j}^{(\alpha)}=\Tr[\varrho P_{j}^{(\alpha)}]$, where $P_j^{(\alpha)}=\ket{e_j^{(\alpha)}}\bra{e_j^{(\alpha)}}$.
For a set of $L$ MEBs that satisfy  $\lambda_{\max}(N)<1/(L-1)$, we have
\begin{eqnarray}\label{Theorem-Eq-EUR1}
\sum_{\alpha=1}^{{L}}\sum_{j=0}^{d-1}\left[p_{j}^{(\alpha)}\right]^2 \le f_L(q)\bigg[\Tr[\varrho^2]-\frac{1}{d}\bigg]+\frac{L}{d},
\end{eqnarray}
where
\begin{eqnarray}\label{fL(q)}
f_{L}(q)=1+(L-1)\lambda_{\max}(N).
\end{eqnarray}
Obviously, $|N_{k,d-k}|=0$ for MUBs for which   $\mu=1/d$, thus leading to the results of \cite{MolmerPRA2009}.
\end{theorem}

\begin{proof}
Building on the method presented in Ref. \cite{MolmerPRA2009},  we define a two-qudit pure state  $[\varrho\otimes \Id]\ket{\phi}$, where
$\varrho$ is an arbitrary quantum state on $\mathcal{H}$ and
\begin{eqnarray}\label{two-qudit-phi}
\ket{\phi}&=&\frac{1}{\sqrt{d}}\sum_{j=0}^{d-1}\ket{e_j^{(1)}}\ket{e_j^{(1)}}^\ast,
\end{eqnarray}
is a  maximally entangled state of  the  composite system $\mathcal{H}\otimes \mathcal{H}$.
One can easily show that $\braket{\phi}{\phi_k^{(\alpha)}}=0$ for $k=1,\ldots, d-1$ and  $\alpha=1,\ldots,L$. Invoking this and Lemma \ref{Lemma-Reciprocal}, vectors $\{\ket{\phi}, \ket{\phi_k^{(\alpha)}}\}$  span a $D$-dimensional subspace of $\mathcal{H}\otimes\mathcal{H}$ where $D=L(d-1)+1$. Defining  $\{\ket{\chi_l}\}_{l=1}^{d^2-D}$ as an orthonormal basis for the orthogonal complement of this subspace, and adding these vectors to the previously defined set of vectors we obtain, in general, a nonorthogonal basis for $\mathcal{H}\otimes \mathcal{H}$ as $\{\ket{\phi},\ket{\phi_k^{(\alpha)}},\ket{\chi_l}\}$ where  $k=1,\ldots, d-1$;  $\alpha=1,\ldots,L$; and $l=1,\ldots,d^2-D$.

By  expanding $[\varrho\otimes \Id]\ket{\phi}$  in terms of the aforementioned basis of $\mathcal{H}\otimes \mathcal{H}$,  we get
\begin{eqnarray}\label{Rho1Vector}
[\varrho\otimes \Id]\ket{\phi}=\frac{1}{d}\ket{\phi}+\sum_{\alpha=1}^{L}\sum_{k=1}^{d-1}[G^{-1}\Omega]_k^{(\alpha)}\ket{\phi_k^{(\alpha)}}+\ket{\chi},
\end{eqnarray}
where the first term follows from $\bra{\phi}[\varrho\otimes \Id]\ket{\phi}=1/d$ and the  unnormalized vector  $\ket{\chi}$ represents the component of $[\varrho\otimes \Id]\ket{\phi}$ that lies in the orthogonal subspace  spanned by $\{\ket{\chi_l}\}$.
For the second term, using  Lemma \ref{Lemma-Reciprocal}, we have $\bra{\widetilde{\phi}_k^{(\alpha)}}[\varrho\otimes \Id]\ket{\phi}=[G^{-1}\Omega]_k^{(\alpha)}$,
where    $G^{-1}$ refers to the inverse of the  Gram matrix $G$  defined by   Eq. \eqref{GramNkkp}. The  $L(d-1)$-dimensional vector $\Omega$ is defined by its components $\Omega_l^{(\beta)}$ as
\begin{eqnarray}\label{Omega-beta-l}
\Omega_l^{(\beta)}=\frac{1}{d}\sum_{j=0}^{d-1}{\omega^\ast}^{lj}p_{j}^{(\beta)},
\end{eqnarray}
for $\beta=1,\ldots,L$ and $l=1,\ldots,d-1$. Here  $p_{j}^{(\beta)}=\bra{{\phi}_j^{(\beta)}}\varrho\ket{{\phi}_j^{(\beta)}}$ is the probability of obtaining the $j$th result when the system is measured on the $\beta$th basis.

For the squared norm of the vector defined by Eq. \eqref{Rho1Vector}, we have
\begin{eqnarray}\label{Tr-rho2}
\bra{\phi}[\varrho^2\otimes \Id]\ket{\phi}&=& \frac{1}{d}\Tr[\varrho^2]  \\ \nonumber
&=& \frac{1}{d^2}+\Omega^\dagger G^{-1} \Omega +\braket{\chi}{\chi}.
\end{eqnarray}
From this  we get
\begin{eqnarray}\label{ineq1}
\frac{1}{d}\Big[\Tr[\varrho^2] -\frac{1}{d}\Big] &=& \Omega^\dagger G^{-1} \Omega +\braket{\chi}{\chi} \\ \nonumber
 &\ge &\Omega^\dagger G^{-1} \Omega \\ \nonumber
 &\ge & \|\Omega\|^2 \lambda_{\min}(G^{-1}),
\end{eqnarray}
where $\lambda_{\min}(G^{-1})$ is the smallest eigenvalue of $G^{-1}$ and the squared norm $\|\Omega\|^2=\Omega^\dagger \Omega$ can be expressed as
\begin{eqnarray}\label{OmegaDaggerOmega}
\Omega^\dagger \Omega &=&\sum_{\alpha=1}^{L}\sum_{k=1}^{d-1}[\Omega_k^{(\alpha)}]^\ast[\Omega_k^{(\alpha)}] \\ \nonumber
&=&\frac{1}{d^2}\sum_{j,j^\prime=0}^{d-1}\sum_{\alpha=1}^{L}\sum_{k=1}^{d-1}\omega^{k(j^\prime-j)}p_{j^\prime}^{(\alpha)}p_{j}^{(\alpha)} \\ \nonumber
&=&\frac{1}{d^2}\sum_{j,j^\prime=0}^{d-1}\sum_{\alpha=1}^{L}[d\delta_{jj^\prime}-1]p_{j^\prime}^{(\alpha)}p_{j}^{(\alpha)}
\\ \nonumber
&=&\frac{1}{d}\sum_{\alpha=1}^{L}\sum_{j=0}^{d-1}[p_{j}^{(\alpha)}]^2-\frac{L}{d^2}.
\end{eqnarray}
By putting  Eq. \eqref{OmegaDaggerOmega} into Eq. \eqref{ineq1}, we get
\begin{eqnarray}\nonumber
\sum_{\alpha=1}^{{L}}\sum_{j=0}^{d-1}\left[p_{j}^{(\alpha)}\right]^2 & \le & \frac{1}{\lambda_{\min}(G^{-1})}\bigg[\Tr[\varrho^2]-\frac{1}{d}\bigg]+\frac{L}{d} \\
&=& \lambda_{\max}(G)\bigg[\Tr[\varrho^2]-\frac{1}{d}\bigg]+\frac{L}{d},
\end{eqnarray}
where the last line follows from the fact that  the Gram matrix of a linearly independent set of vectors  is positive definite,  as such  $\lambda_{\max}(G)=1/\lambda_{\min}(G^{-1})$.
To complete the proof, it remains only to determine  the eigenvalues of the Gram matrix $G$. A complete derivation of the eigenvalues is provided in Appendix \ref{Appendix-Eigenvalues-G}, from which the maximum eigenvalue we need here is given by $\lambda_{\max}(G)=f_L(q)$, where $f_L(q)$ is defined by Eq. \eqref{fL(q)}. This concludes  the proof.
\end{proof}

In proving the theorem, we observe two different features of the $d(L-1)$-dimensional vector $\Omega$ and the $(d-1)\times(d-1)$ matrix $N$: The  vector $\Omega$ contains the probabilities of obtaining different outcomes when the system, initially prepared in state $\varrho$, is measured in MEBs, while the Gram matrix $G$, and thus  the  submatrix $N$, provides all information on the MEBs.   This last feature enables us to parameterize  the probability distribution  $q=(q_0,\ldots,q_{d-1})$ associated with  each MEB, using  $d-1$ independent parameters provided by the entries of the matrix $N$, as $N_{k,d-k}=\delta_{2k-1}+i\delta_{2k}$, where  $(d-1)$ real  parameters $\delta_{1},\ldots,\delta_{d-1}$ are defined  by Eqs. \eqref{delta-odd} and \eqref{delta-even}, respectively (see Sec.  \ref{Section-parameterization}).

Quantum incompatibility  imposes limitation on  the information that can be extracted from measurements of different observables, and the uncertainty relation quantifies how the incompatibility of measurements restricts the precision with which we can simultaneously obtain information about incompatible  observables. The entropic uncertainty relation provides a limit on how much information one can gain from  measuring incompatible measurements.
The inequality given by Theorem \ref{Theorem-EUR1} for probabilities of projective measurement in MEBs can be used to construct an  entropic uncertainty inequality.

\begin{theorem}\label{Theorem-EUR2}
With the same notation as Theorem \ref{Theorem-EUR1}, suppose there exist $L$ MEBs in a  $d$-dimensional Hilbert space $\mathcal{H}$ and that the probability of obtaining the $j$th result when the system, initially prepared in state $\varrho$, is measured on the $\alpha$th basis is given by $p_{j}^{(\alpha)}=\Tr[\varrho P_{j}^{(\alpha)}]$. Then by defining $H(\mathcal{B}^{(\alpha)}|\varrho)=-\sum_{j=0}^{d-1}p_{j}^{(\alpha)}\log{p_{j}^{(\alpha)}}$ as the Shannon entropy of the probability distribution $\{p_{j}^{(\alpha)}\}_{j=0}^{d-1}$, we have
\begin{eqnarray}\nonumber
\sum_{\alpha=1}^{{L}}H(\mathcal{B}^{(\alpha)}|\varrho) &\ge & [L-KC](K+1)\log{(K+1)} \\
&-&[L-(K+1)C]K\log{K},
\end{eqnarray}
where $K=\lfloor \frac{L}{C} \rfloor$, and $C$ is the upper  bound for $\sum_{\alpha=1}^{{L}}\sum_{j=0}^{d-1}\big[p_{j}^{(\alpha)}\big]^2$. Here,  for a set of $L$ MEBs that satisfy the condition   $\lambda_{\max}(N)<1/(L-1)$,  $C=f_L(q)\big[\Tr[\varrho^2]-\frac{1}{d}\big]+\frac{L}{d}$ is the upper bound given in Eq. \eqref{Theorem-Eq-EUR1}.
\end{theorem}
\begin{proof}
The proof follows exactly the same steps as   Ref. \cite{MolmerPRA2009}, where the result was established for MUBs based on an inequality by Harremo\"{e}s and Tops{\o}e \cite{HarremosIEEE2001}. Since our proof is identical in structure except for the use of a different  upper  bound $C$ for $\sum_{\alpha=1}^{{L}}\sum_{j=0}^{d-1}\big[p_{j}^{(\alpha)}\big]^2$, we omit the details here and refer the reader to Ref. \cite{MolmerPRA2009} for the complete argument.
\end{proof}

\section{Entanglement detection}\label{Sec-EntDetection}
A powerful tool for detecting entanglement is the use of positive but not completely positive maps.
A linear  map $\Phi:\mathcal{B}(\mathcal{H}_1)\rightarrow \mathcal{B}(\mathcal{H}_2)$ is called positive if $\Phi X\ge 0\; \forall X\ge 0$, where $\mathcal{B}(\mathcal{H})$ denotes the space of bounded operators acting on the Hilbert space  $\mathcal{H}$.
The map $\Phi$ is called completely positive if and only if  its trivial extension $\mathrm{id}_n\otimes \Phi$ also remains positive for all dimension $n$.
For a given positive but not completely positive  map $\Phi$, one can construct the corresponding entanglement witness $W=(d-1)(\textrm{id}_n\otimes \Phi)\ket{\phi_d^{+}}\bra{\phi_d^{+}}$, where $\ket{\phi_d^{+}}=\frac{1}{\sqrt{d}}\sum_{i=0}^{d}\ket{ii}$. It turns out that $W$ has nonnegative expectation value on any separable state $\rho$, i.e. $\Tr{\rho W}\ge 0$  for all separable states. Accordingly, a state for which $\Tr{\rho W}<0$ is entangled. For a given set of $L$ MUBs,  the authors of \cite{ChruscinskiPRA2018}  introduced the   positive and trace-preserving  map
\begin{equation}\label{Phi-map-MUB}
\phi X = \phi_{*}X - \dfrac{1}{d-1}\sum_{\alpha=1}^{{L}}\sum_{k,l=0}^{d-1}\mathcal{O}_{k,l}^{(\alpha)}\Tr(\tilde{X}P_{l}^{(\alpha)})P_{k}^{(\alpha)},
\end{equation}
where   $P_k^{(\alpha)}$, for $k=0,\ldots,d-1$, forms a projective measurement on the $\alpha$th basis,
 $\phi_{*}X=\dfrac{1}{d}\Id \Tr X$ defines the completely depolarizing channel and  $\tilde{X}=X-\phi_{*}X$ denotes the traceless part of X. Also, $\mathcal{O}^{(b)}$ is a set of orthogonal rotation in $\mathbb{R}^d$ around the axis $\boldsymbol{n}_\ast=(1,1,\ldots,1)/\sqrt{d}$ so that $\mathcal{O}^{(b)}\boldsymbol{n}_\ast=\boldsymbol{n}_\ast$.
The corresponding entanglement witness reads  \cite{ChruscinskiPRA2018}
\begin{eqnarray}\label{WitnessE}
W= \dfrac{d +L-1}{d} \Id \otimes \Id - \sum_{\alpha=1}^{{L}}\sum_{k,l=0}^{d-1}\mathcal{O}_{kl}^{(\alpha)} \overline{P}_{l}^{(\alpha)}\otimes P_{k}^{(\alpha)},
\end{eqnarray}
where $\overline{P}_{l}^{(\alpha)}$ is the conjugation of $P_{l}^{(\alpha)}$ and $1\le L \le d+1$.
The following proposition presents a  positive   map in terms of a set of $L$ MEBs for a $d$-dimensional system.

\begin{proposition}
Suppose   $P_k^{(\alpha)}=\ket{e_k^{(\alpha)}}\bra{e_k^{(\alpha)}}$ for $k=0,\ldots,d-1$ forms a projective measurement on the $\alpha$th basis of a given set of $L$ MEBs. For a set of $L$ MEBs that satisfy   $\lambda_{\max}(N)<1/(L-1)$, the map
\begin{eqnarray}\label{Phi-map-MEB}
\Phi X &=&\left[\frac{d+\gamma L-1}{d(d-1)}\right]\Tr[X]\Id_d \\ \nonumber
&-& \frac{\gamma}{d-1}\sum_{\alpha=1}^{{L}}\sum_{k,l=0}^{d-1}\mathcal{O}_{k,l}^{(\alpha)}\Tr[X P_{l}^{(\alpha)}]P_{k}^{(\alpha)},
\end{eqnarray}
is positive and trace preserving.
Her  $\gamma$ is given by
\begin{eqnarray}\label{gamma}
\gamma=\bigg[\frac{d-1}{(d-1)f_L(q)+L(L-1)(dq_{\max}-1)}\bigg]^{1/2},
\end{eqnarray}
where $q_{\max}=\max_{k}\{q_0,\ldots,q_{d-1}\}$ and $f_L(q)$   is defined by Eq. \eqref{fL(q)}.
Also, $\mathcal{O}^{(b)}$ is a set of orthogonal rotation in $\mathbb{R}^d$ around the axis $\boldsymbol{n}_\ast=(1,1,\ldots,1)/\sqrt{d}$ so that $\mathcal{O}^{(b)}\boldsymbol{n}_\ast=\boldsymbol{n}_\ast$.
\end{proposition}
\begin{proof}
First note that the map $\Phi$, defined by Eq.  \eqref{Phi-map-MEB},
is trace preserving, i.e., $\Tr[\Phi X]=\Tr[X]$.
To show that it is also positive, we require that $\Phi X\ge 0$ for all $X\ge 0$. According to  Mehta's lemma, a Hermitian matrix $A$ is positive if $\Tr[A^2]\le \Tr[A]/(d-1)$.
In view of this, the map $\Phi$ is positive if and only if $\Tr[\Phi \varrho ]^2\le \Tr[\varrho]/(d-1)$ for all rank-1 projection operators $P=\ket{\psi}\bra{\psi}$.
After some tedious but straightforward calculations one finds that
\begin{eqnarray}
\Tr[\Phi P]^2=\left[\frac{(d-1)^2-\gamma^2L^2}{d(d-1)^2}\right]+\frac{\gamma^2}{(d-1)^2}T,
\end{eqnarray}
where the expression  $T$ is defined as
\begin{eqnarray}
T&=& \sum_{\alpha=1}^{{L}}\sum_{l=0}^{d-1}\left\{\Tr[P P_{l}^{(\alpha)}]\right\}^2 \\ \nonumber
&+&\sum_{\alpha\ne \alpha^\prime}^{{L}}\sum_{k,l=0}^{d-1}\sum_{k^\prime,l^\prime=0}^{d-1}\mathcal{O}_{k,l}^{(\alpha)}\mathcal{O}_{k^\prime,l^\prime}^{(\alpha^\prime)}\Tr[P P_{l}^{(\alpha)}]\Tr[P P_{l^\prime}^{(\alpha^\prime)}]q_{k\oplus k^\prime}.
\end{eqnarray}
Remarkably, for a set of $L$ MUBs, the second term of the above expression reduces to $L(L-1)/d$ and the first term is bounded from above by $1+(L-1)/d$ so that the  condition $\Tr[\Phi P]^2\le 1/(d-1)$ is satisfied for $\gamma=1$. For  MEBs, on the other hand, the second term is bounded from above by $L(L-1)q_{\max}$ where $q_{\max}=\max_{k}\{q_k\}$. For the fist term, however, the upper bound is given by
\begin{equation}
\sum_{\alpha=1}^{{L}}\sum_{l=0}^{d-1}\left\{\Tr[P_{\psi} P_{l}^{(\alpha)}]\right\}^2  \le \frac{(d-1)f_L(q)+L}{d},
\end{equation}
where $f_L(q)$ is defined by Eq. \eqref{fL(q)} (see Theorem \ref{Theorem-EUR1} and its proof). Putting everything together, one can find that the map $\Phi$ is positive if and only if $\gamma$ is given by Eq. \eqref{gamma}.
This completes the proof that the map $\Phi$ is positive and trace preserving.
\end{proof}

Corresponding to the map $\Phi$, defined by Eq. \eqref{Phi-map-MEB}, one can define the entanglement witness as 
\begin{eqnarray}\label{WitnessMEB}
W= \dfrac{d +\gamma L-1}{d} \Id \otimes \Id - \gamma \sum_{\alpha=1}^{{L}}\sum_{k,l=0}^{d-1}\mathcal{O}_{kl}^{(\alpha)} \overline{P}_{l}^{(\alpha)}\otimes P_{k}^{(\alpha)}.
\end{eqnarray}
For illustration, we consider an example of isotropic states of a $d\times d$ system, defined by
\begin{equation}\label{isotropicState}
\varrho_{I}=\omega\ket{\phi^{+}_d}\bra{\phi^{+}_d}+\frac{1-\omega}{d^2}\Id_d\otimes \Id_d,
\end{equation}
 where $0\le \omega \le 1$ and  $\ket{\phi_d^{+}}=\frac{1}{\sqrt{d}}\sum_{i=0}^{d}\ket{ii}$. It is shown  that isotropic states are entangled if and only if $\omega > 1/(d+1)$. Straightforward calculations show that
 \begin{eqnarray}\nonumber
\Tr[\varrho_{I}W] &=& \frac{d+L \gamma  \omega -1}{d}-\frac{\gamma \omega}{d}\sum_{\alpha=1}^{L}\sum_{k,l=0}^{d-1}\mathcal{O}_{kl}^{(\alpha)}\Tr[P_l^{(\alpha)}P_{k}^{(\alpha)}] \\ \nonumber
&=& \frac{d+L \gamma  \omega -1}{d}-\frac{\gamma \omega}{d}\sum_{\alpha=1}^{L}\Tr[\mathcal{O}^{(\alpha)}] \\ \nonumber
&=&\frac{(d-1)}{d}[1-L \gamma\omega],
 \end{eqnarray}
where in the the second  equality we have used Eq. \eqref{MEB1}, and the third  equality follows from the fact that the optimum detection happens when $\Tr[\mathcal{O}^{(\alpha)}]$ takes its maximum value $d$, i.e., when $\mathcal{O}^{(\alpha)}=\Id_d$ for all $\alpha$. It turns out  that $W$ can detect the entanglement of  isotropic states when  $\omega >\frac{1}{L \gamma }$. Accordingly,  entanglement detection depends on the product of two parameters $L$ and $\gamma$:  how many pairs of MEBs are used to detect entanglement and how effective the measurements are. This provides a tradeoff relation between the number of measurements and their  effectiveness in the sense that the more incompatible the measurements are,  the fewer measurements are needed for detection. Clearly, complete entanglement detection is possible when a full set of $d+1$ measurements on MUBs is performed.

Figure \ref{WitnessGamma}  illustrates this quantity-quality tradeoff in terms of the index  $\mu$  for $d=3$ and different values of   $L=2,3,4$. In this case, $\lambda_{\max}(N)=\sqrt{(3\mu-1)/2}$; as such,  the constraint \eqref{Constraint-lambdaN} leads to $\mu<(2/(L-1)^2+1)/3$. This yields the bounds $\mu<1$, $1/2$, and $11/27$, for $L=2$, $3$, and $4$, respectively.   In all cases, the solid curves represent $(L\gamma)^{-1}$  when MEBs are used to construct witnesses. The dashed lines, on the other hand, capture the bounds when a set of $L$ MUBs is used to construct witnesses so that $\gamma=1$. Not surprisingly, for a fixed $L$ the MUB-based witness is finer than the MEB-based witness; however, if a MEB-based witness employs a larger set of measurements, it can surpass the MUB-based witness in detecting entanglement, particularly within the specific index  $\mu$ of the underlying distribution $q$ of  the MEBs.

\begin{figure}[h]
\includegraphics[scale=0.95]{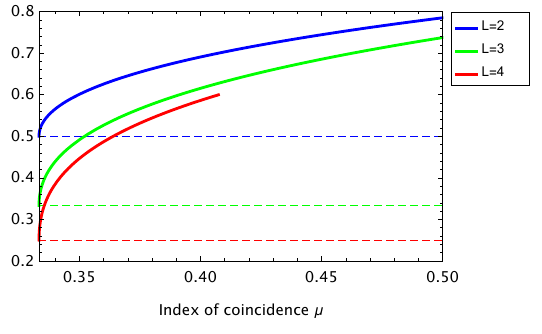}
\caption{Plot of $(L\gamma)^{-1}$ in terms of index $\mu$ for $d=3$ and different values of $L=2,3,4$. For equal $L$'s, the MUB-based witness (dashed lines) always captures a wider range of states relative to the MEB-based witness (solid curves). However, if a MEB-based witness employs a larger set of measurements, it can surpass the MUB-based witness in detecting entanglement, particularly within low  values of the index $\mu$. As discussed in the  text, for $L=4$, the plot is restricted to $\mu\in[1/3,11/27)$, due to the imposed limit on $\mu$.  }
\label{WitnessGamma}
\end{figure}

The inequality $\lambda_{\max}(N)<1/(L-1)$ imposes a tradeoff relation on the MEBs, in the sense that for a given set of MEBs associated with a probability distribution $q=(q_0,\ldots,q_{d-1})$, only $L<1/\lambda_{\max}(N)+1$ of them are relevant  for entanglement detection, even if more than $L$ could  be constructed. The constraint does not make sense for MUBs, as $N_{k,d-k}=0$ for the uniform distribution $q=(1/d,\ldots,1/d)$, so they are already optimal. For MEBs, however, it arises   from the possibility of linear  dependence among  the maximally entangled states in Eq. \eqref{two-qudit-phi}, which renders certain bases redundant for the entanglement detection. In view of this,   the  significance of MEBs is pronounced when their distribution is sufficiently  close  to the uniform distribution. Despite the constraint imposed on the MEBs, they may still be valuable in certain  dimensions, such as $d=6$, where the maximum number of three  MUBs has been identified. This raises the question of whether it is possible for a set of at least four  MEBs to exist in dimension $d=6$ and, if such a set is found, can it be sufficiently close to the uniform distribution such that  the constraint does not exclude any of the bases. This possibility merits further investigation.

\section{Conclusion}\label{Sec-Conclusion}
In this work we  introduced the concept of mutually equibiased bases, a generalization of the mutually unbiased bases. For such bases, we relaxed  the condition of mutual  maximum unpredictability  that holds for MUBs; instead, we associated   a probability distribution $q=(q_0,\ldots,q_{d-1})$ with  each of the MEBs. This distribution captures  the predictability of the outcomes of a measurement in one basis when the system is prepared in a state from  the other basis. However, the information extracted from the measurement is unbiased, meaning that  it cannot reveal which  state  from  which  basis   the system is prepared in and  which basis is used for the measurement.

We derived a  set of  $d+1$ MEBs for $d=2$ and $d=3$   and found that  although  for $d=2$   a complete set of three  MEBs exist for all values of the index $\mu=\sum_{k=0}^{d-1}q_k^2$, for $d\ge 3$ this  is not the case. More precisely,  for  $d=3$, the constraints imposed by the mutual equibiasedness limits the range of the index $\mu$ to the feasible region  $1/3\le \mu \le 1/2$ within which a set of four MEBs exists. Furthermore, in addition to this restriction on $\mu$, for a given $\mu$ in the feasible region the components of the probability distribution $q=(q_0,q_1,q_2)$ are also limited to some specific discrete points; only one probability distribution $q$, up to permutation of its components, can generate a  set of four MEBs.

For a qudit system, we derived an inequality  for  the probabilities of projective measurements in MEBs. We further used the inequality  to derive   an  entropic uncertainty inequality in MEBs. As expected, the maximal incompatibility inherent to  MUBs breaks down in the case of MEBs of the same size.
The inequality  for  probabilities is then used to obtain a class of positive maps, from which we constructed  entanglement witnesses. Using the MEBs for entanglement detection faces a constraint; depending on the underlying probability distribution $q$, the number of bases that can be employed from a set of $L$ MEBs is subject to a limit. This reveals  a key  distinction between MEBs and MUBs in such applications.

We examined our entanglement witness for $d\times d$  isotropic states and, in the case $d=3$, we evaluated it using a set of $L$ MEBs; we presented the results through illustrative plots for $L=2,3,4$.
We showed that in the absence of a complete set of MUBs, having a set of MEBs with at least one additional basis can  enhance the construction of finer entanglement witnesses.

These results provide a deeper understanding of the structure and properties of MEBs, offering new tools for entanglement detection and quantum measurement theory.
The challenging dimension  $d=6$ is currently under investigation to determine whether there exist values of the index $\mu$ for which one can construct a set of at least four MEBs, exceeding a set of  three MUBs currently known for this dimension.  We hope that this work will serve as a stepping stone for further research, ultimately advancing our understanding of MUBs. In particular, for dimensions  where no more than three MUBs are known, even a set of four  MEBs with $\mu > 1/d$ could shed light on the challenging open problem of constructing complete set of MUBs.

\section*{acknowledgment}
This work was based upon research funded by Iran National Science Foundation  under Project No. 4021765.


\appendix

\section{Proof of Proposition \ref{Proposition-d3} }\label{Appendix-ProofProposition}
\begin{proof}
Having the first unitary matrix $\mathcal{U}$, we have to find two other unitary matrices $\mathcal{V}$ and $\mathcal{W}$ with the same moduli of  matrix elements as $\mathcal{U}$. The second unitary matrix  $\mathcal{V}$ can be constructed  as $\mathcal{V}=\mathcal{D}_{\alpha}\mathcal{U}$ where $\mathcal{D}_{\alpha}=\diag\{1,\e^{i\alpha_1},\e^{i\alpha_2}\}$.
Regarding   the probability $q=(q_0,q_1,q_2)$,  although both $\mathcal{U}$ and $\mathcal{V}$ generate bases that are equibiased with respect to the standard basis, they are not necessarily equibiased with respect to each other. For this to hold, $\mathcal{U}$ and $\mathcal{V}$ must be mutually equibiased, i.e., the unitary matrix $\mathcal{V}^\dagger\mathcal{U}$  must generate the same probability matrix $\mathcal{Q}$, up to a permutation of its entries.
This places some constraints on phases $\alpha_1$ and $\alpha_2$.

To proceed with this claim, we have to compute the matrix $\mathcal{V}^\dagger\mathcal{U}$ and require  that the squared moduli of its diagonal entries generate the probability  distribution $q^\sigma=(q_{\sigma(0)},q_{\sigma(1)},q_{\sigma(2)})$, such that
\begin{eqnarray}\label{|VdaggerU|2}
|[\mathcal{V}^\dagger\mathcal{U}]_{00}|^2=q_{\sigma(0)},\;  |[\mathcal{V}^\dagger\mathcal{U}]_{11}|^2=q_{\sigma(2)},\;  |[\mathcal{V}^\dagger\mathcal{U}]_{22}|^2=q_{\sigma(1)}.
\end{eqnarray}
Here $\sigma\in S_3$ denotes a permutation of three elements $\{0,1,2\}$, where  $S_3$ is the symmetric group of degree $3$.  For clarity,  the six  operations of $S_3$  can be written  as  the identity $\mathrm{id}$, the transpositions $(01)$, $(02)$, $(12)$, and the forward and backward cycles $(012)$ and $(210)$.
After some calculations, we get
\begin{eqnarray}\label{3constraints-1}
\mu+2q_0q_1\cos{\alpha_1}+2q_0q_2\cos{\alpha_2}+2q_1q_2\cos{\alpha_{12}}=q_{\sigma(0)}, \\ \label{3constraints-2}
\mu+2q_1q_2\cos{\alpha_1}+2q_0q_1\cos{\alpha_2}+2q_0q_2\cos{\alpha_{12}}=q_{\sigma(2)}, \\ \label{3constraints-3}
\mu+2q_0q_2\cos{\alpha_1}+2q_1q_2\cos{\alpha_2}+2q_0q_1\cos{\alpha_{12}}=q_{\sigma(1)},
\end{eqnarray}
where we have defined $\alpha_{12}=\alpha_1-\alpha_2$, for the sake of brevity.
Summing the three equations above, we obtain
\begin{equation}\label{AlphaCondition}
\cos{\alpha_1}+\cos{\alpha_2}+\cos{\alpha_{12}}=\frac{1-3\mu}{1-\mu},
\end{equation}
which holds for  $\mu\in[1/3,5/9]$.

Equation  \eqref{AlphaCondition} imposes  a constraint on the allowed values of the phases $\alpha_1$ and $\alpha_2$ under which two unitary matrices $\mathcal{U}$ and $\mathcal{V}$ are mutually equibiased. A contour plot of this equation in $\alpha_1\alpha_2$-plane for different values of  $\mu$    is illustrated in Fig. \ref{FIG-AlphaCondition}. The contours that are bounded within the  hexagon correspond to $\mu\in [1/3, 1/2]$ and those outside the hexagon, i.e., within the two outer triangles,  correspond to $\mu\in (1/2,5/9]$.

\begin{figure}[t]
\includegraphics[scale=0.8]{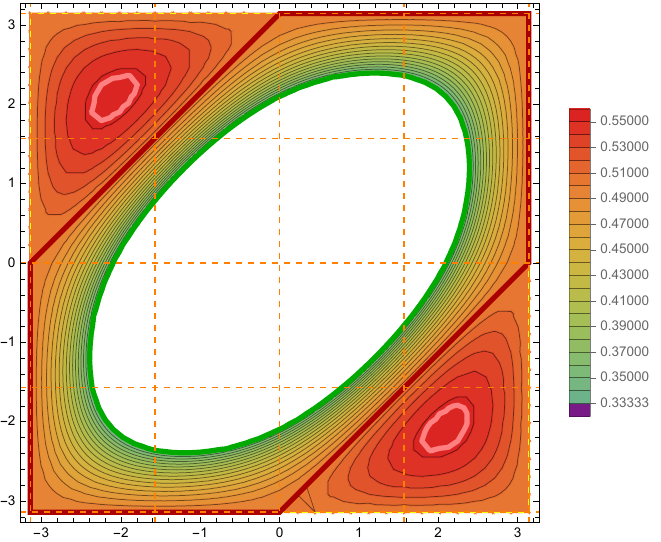}
\caption{Contour plot of the left-hand side  of Eq. \eqref{AlphaCondition} in $\alpha_1\alpha_2$ plane for  $\alpha_1,\alpha_2\in(-\pi,+\pi)$.  (a) The inner boundary (green ellipse) corresponds to $\mu=1/3$, (b) the outer boundary (red hexagon) corresponds to $\mu=1/2$, and (c) two little pink contours are plotted  for  $\mu\approx5/9$. The exact value for $\mu=5/9$ correspond to two points  at $(2\pi/3,-2\pi/3)$ and $(-2\pi/3,2\pi/3)$.   The $\alpha_1\alpha_2$ plane is divided by dashed grid lines to make it easier to identify the quadrant of the angles $\alpha_1$ and $\alpha_2$.  }
\label{FIG-AlphaCondition}
\end{figure}

Clearly, the constraint imposed by Eq. \eqref{AlphaCondition} is invariant under the symmetric group $S_3$; it  depends  only on  $\mu$ but not on where  the probability $q=(q_0,q_1,q_2)$ is located in the feasible region.  However, a look at the left-hand sides  of Eqs. \eqref{3constraints-1}-\eqref{3constraints-3} reveals that this set is  invariant  under a subgroup of $S_3$, namely, the cyclic group $C_3=\{\mathrm{id},(012),(210)\}$.
The set, however, will possess the full symmetry under $S_3$ if and only if two out of the three terms $\cos{\alpha_1}$, $\cos{\alpha_2}$, and $\cos{\alpha_{12}}$ become  equal. This results in three possible cases as
\begin{enumerate}
  \item[(a) : ] $\alpha_1=0,\quad\quad  \alpha_2=\theta^{\mu}$,
  \item[(b) : ] $\alpha_1=\theta^{\mu},\quad\; \alpha_2=0$,
  \item[(c) : ] $\alpha_1=\theta^{\mu}, \quad\; \alpha_2=\theta^{\mu}$,
\end{enumerate}
where, from Eq. \eqref{AlphaCondition}, $\theta^{\mu}$ is obtained as Eq. \eqref{theta-mu2}. The index condition  $\mu\le 1/2$ is obtained from  $|\cos{\theta^{\mu}}|\le 1$.   Using this in Eqs. \eqref{3constraints-1}-\eqref{3constraints-3}, we get for  case (a) the  set of equations
\begin{eqnarray}
q_{\sigma(0)}=\frac{2q_0q_1}{1-\mu}, \quad q_{\sigma(2)}=\frac{2q_1q_2}{1-\mu},\quad q_{\sigma(1)}=\frac{2q_0q_2}{1-\mu}.
\end{eqnarray}
Depending on the choice of the permutation $\sigma\in S_3$, the   set of  equations determines the probability distribution $q=(q_0,q_1,q_2)$. More precisely, the choice   $\sigma=\mathrm{id}$ fixes $q_1=q^{\mu}_0$, so that $q=(q^{\mu}_{\pm},q^{\mu}_0,q^{\mu}_{\mp})$, where $q^{\mu}_0=(1-\mu)/2$ and $q^{\mu}_{\pm}$ is given in Eq. \eqref{qpm-mu}.      Similarly, the  cycle  $\sigma=(012)$  fixes $q_0=q^{\mu}_0$, giving  $q=(q^{\mu}_0,q^{\mu}_{\pm},q^{\mu}_{\mp})$, while the cycle  $\sigma=(210)$ yields  $q=(q^{\mu}_{\pm},q^{\mu}_{\mp},q^{\mu}_0)$. In the case of  the transpositions  $(01)$, $(02)$, and $(12)$,  the set possesses  solution only for $\mu=1/3$, corresponding to the uniform distribution $q=(1/3,1/3,1/3)$.

In a similar manner,  for cases (b) and (c), the set of equations \eqref{3constraints-1}-\eqref{3constraints-3} reduce to
\begin{eqnarray}
q_{\sigma(0)}=\frac{2q_0q_2}{1-\mu}, \quad q_{\sigma(2)}=\frac{2q_0q_1}{1-\mu},\quad q_{\sigma(1)}=\frac{2q_1q_2}{1-\mu},
\end{eqnarray}
and
\begin{eqnarray}
q_{\sigma(0)}=\frac{2q_1q_2}{1-\mu}, \quad q_{\sigma(2)}=\frac{2q_0q_2}{1-\mu},\quad q_{\sigma(1)}=\frac{2q_0q_1}{1-\mu},
\end{eqnarray}
respectively. In each case, the corresponding set of equations gives, up to a permutation,  the same set of solutions  \eqref{q0-mu} and \eqref{qpm-mu} for the probabilities $(q_0,q_1,q_2)$.  Irrespective of which case (a), (b), or (c) is considered, only three out of six possible permutations $\sigma$, namely, those corresponding to the  cyclic group $C_3=\{\mathrm{id},(012),(210)\}$,  are  contributed in the solutions. Table \ref{Table-Alpha1} presents all possible solutions of Eqs. \eqref{3constraints-1}-\eqref{3constraints-3}, corresponding to the three cases (a), (b), and (c) and three permutations $\mathrm{id}$, $(012)$,  and $(210)$.

\begin{table}[h!]
  \begin{center}
    \caption{Possible solutions of Eqs. \eqref{3constraints-1}-\eqref{3constraints-3}.}\label{Table-Alpha1}
    \begin{tabular}{|l|c|c|c|} \cline{2-4}
     \multicolumn{1}{@{}c|}{}      &
 $\sigma=\mathrm{id}$ & $\sigma=(012)$  & $\sigma=(210)$  \\   \hline
$\alpha=(0,\theta^{\mu})$ & $q=(q^{\mu}_{\pm},q^{\mu}_0,q^{\mu}_{\mp})$ & $q=(q^{\mu}_0,q^{\mu}_{\pm},q^{\mu}_{\mp})$  &    $q=(q^{\mu}_{\pm},q^{\mu}_{\mp},q^{\mu}_0)$   \\ \hline
$\alpha=(\theta^{\mu},0)$  &  $q=(q^{\mu}_{\pm},q^{\mu}_{\mp},q^{\mu}_0)$ & $q=(q^{\mu}_{\pm},q^{\mu}_0,q^{\mu}_{\mp})$  &   $q=(q^{\mu}_0,q^{\mu}_{\pm},q^{\mu}_{\mp})$    \\ \hline
$\alpha=(\theta^{\mu},\theta^{\mu})$ & $q=(q^{\mu}_0,q^{\mu}_{\pm},q^{\mu}_{\mp})$ & $q=(q^{\mu}_{\pm},q^{\mu}_{\mp},q^{\mu}_0)$  &    $q=(q^{\mu}_{\pm},q^{\mu}_0,q^{\mu}_{\mp})$   \\ \hline
    \end{tabular}
  \end{center}
\end{table}

To complete the characterization of the set of MEBs, we have to construct the third unitary matrix $\mathcal{W}$ as $\mathcal{W}=\mathcal{D}_{\beta}\mathcal{U}$ where
$\mathcal{D}_{\beta}=\diag\{1,\e^{i\beta_1},\e^{i\beta_2}\}$. Similar to the arguments given above, we require that the unitaries  $\mathcal{W}\mathcal{U}$ and $\mathcal{W}\mathcal{V}$ generate the unistochastic matrix $\mathcal{Q}$, up to a permutation of its rows and columns. The first requirement is achieved if we select  the pair $(\beta_1, \beta_2)$ from   case (a), (b), or (c),  mentioned above for $(\alpha_1, \alpha_2)$. To ensure  that $\mathcal{W}\mathcal{V}$ also generates the unistochastic matrix $\mathcal{Q}$, we have to choose the pair  $(\beta_1, \beta_2)$ distinct  from the pair $(\alpha_1, \alpha_2)$.
\end{proof}

\begin{table}[h!]
  \begin{center}
    \caption{Correspondence between  possible choices for the distributions $q=(q_0,q_1,q_2)$ (first column), possible choices for the pairs $(\alpha_{1},\alpha_{2})$ and $(\beta_{1}, \beta_{2})$ (first row), and the corresponding three equivalent distributions  for which three unitary matrices   $\mathcal{V}^\dagger \mathcal{U}$, $\mathcal{W}^\dagger \mathcal{U}$, and $\mathcal{W}^\dagger \mathcal{V}$ generate equivalent unistochastic matrices.  }
    \label{Table-AlphaBeta1}
    \begin{tabular}{|c|c|c|c|} \cline{2-4}
     \multicolumn{1}{@{}c|}{}     &
$\begin{array}{ll}
\alpha=(0, \theta^{\mu}) \\
\beta=(\theta^{\mu}, 0)
\end{array}$
 &
$\begin{array}{ll}
\alpha=(0, \theta^{\mu}) \\
\beta=(\theta^{\mu}, \theta^{\mu})
\end{array}$
 &  $\begin{array}{ll}
\alpha=(\theta^{\mu},0) \\
\beta=(\theta^{\mu}, \theta^{\mu})
\end{array}$ \\   \hline
$(q^{\mu}_{+},q^{\mu}_{0},q^{\mu}_{-})$ & $q^{},q^{(012)},q^{(210)}$ & $q^{},q^{(210)},q^{(012)}$  &    $q^{(012)},q^{(210)},q^{}$   \\ \hline
$(q^{\mu}_{+},q^{\mu}_{-},q^{\mu}_{0})$  &  $q^{(210)},q^{},q^{(012)}$ & $q^{(210)},q^{(012)},q^{}$   &   $q^{},q^{(012)},q^{(210)}$    \\ \hline
$(q^{\mu}_{0},q^{\mu}_{+},q^{\mu}_{-})$ & $q^{(012)},q^{(210)},q^{}$ & $q^{(012)},q^{},q^{(210)}$  &    $q^{(210)},q^{},q^{(012)}$   \\ \hline
$(q^{\mu}_{0},q^{\mu}_{-},q^{\mu}_{+})$  &  $q^{(012)},q^{(210)},q^{}$ & $q^{(012)},q^{},q^{(210)}$   &   $q^{(210)},q^{},q^{(012)}$    \\  \hline
$(q^{\mu}_{-},q^{\mu}_{+},q^{\mu}_{0})$ & $q^{(210)},q^{},q^{(012)}$ & $q^{(210)},q^{(012)},q^{}$  &    $q^{},q^{(012)},q^{(210)}$   \\ \hline
$(q^{\mu}_{-},q^{\mu}_{0},q^{\mu}_{+})$  &  $q,q^{(012)},q^{(210)}$ & $q,q^{(210)},q^{(012)}$   &   $q^{(012)},q^{(210)},q$    \\ \hline
    \end{tabular}
  \end{center}
\end{table}

As we have seen, the mutual equibiasedness imposes some constraints on the  index $\mu=q_0^2+q_1^2+q_2^2$  of the distribution $q=(q_0,q_1,q_2)$, meaning that, not any value of $\mu$  is allowed to generate a bistochastic matrix from  the unitary matrix $\mathcal{V}^\dagger\mathcal{U}$. The following remark clarifies  how the feasible region is derived.

\begin{remark}\label{Remark-FeasibleRegion}
The feasible region within the probability triangle is determined by applying a sequence of constraints arising from physical considerations. These conditions successively restrict the parameter space, resulting in a subset that satisfies all necessary properties. Figure  \ref{FIG-Triangle1} outlines how the feasible region is constructed in this case,  with each panel  depicting a specific constraint imposed on the probability  triangle,  as detailed below. \\
Figure \ref{fig:a} shows a  subset of the  probability triangle that satisfies the chain-link conditions \eqref{L-inequalities}.  Corresponding to each probability distribution $q=(q_0,q_1,q_2)$ within the subset, the  associated unitary matrix $\mathcal{U}$  generates the corresponding unistochastic matrix $\mathcal{Q}$. \\
 For $\mu\in[1/3,1)$, the right-hand side of  Eq. \eqref{AlphaCondition} can take any nonpositive  value $(-\infty,0]$, but the left-hand side  is restricted to the interval $[-3/2,3]$. This poses restrictions on the index $\mu$ in the sense that  only distributions with index  $\mu\in[1/3,5/9]$ can generate unistochastic matrices from $\mathcal{V}^\dagger\mathcal{U}$. The resulting feasible region is depicted in  Fig. \ref{fig:b}.\\
Equations   \eqref{3constraints-1}-\eqref{3constraints-3} pose some chain-link conditions on the four segments $q_0$, $q_1$, $q_2$, and $\sqrt{q_i}$ as well, i.e., the nontrivial  conditions $q_i\le q_j+q_k+\sqrt{q_j}$ hold, for $(i,j,k)$ taken from $(0,1,2)$. After imposing  these  constraints, the resulting feasible domain   is reduced to the region shown in   Fig. \ref{fig:c}.\\
Finally,   Eq. \eqref{theta-mu2}, together with the condition $|\cos{\theta^{\mu}}|\le 1$, leads to  the constraint $\mu\le 1/2$.   Figure \ref{fig:d} illustrates the  feasible region corresponding to $\mu\in [1/3,1/2]$, consistent with all previously imposed constraints.
\end{remark}

\begin{figure}[t]
  \centering
\begin{subfigure}[b]{0.23\textwidth}
    \centering
    \includegraphics[scale=0.45]{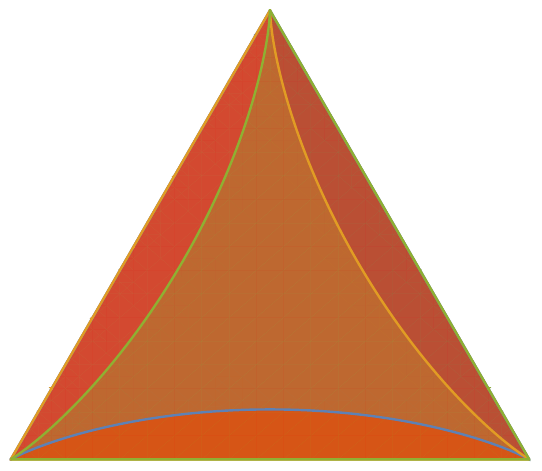}
   \caption{}
    \label{fig:a}
  \end{subfigure}
  \hfill
  \begin{subfigure}[b]{0.23\textwidth}
    \centering
    \includegraphics[scale=0.45]{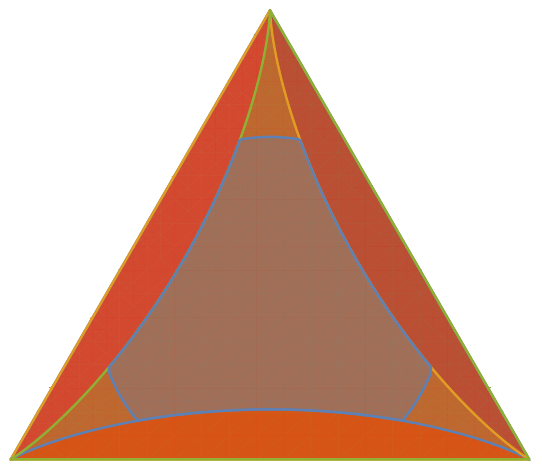}
   \caption{}
    \label{fig:b}
  \end{subfigure}
  \vspace{0.5cm}
  \begin{subfigure}[b]{0.23\textwidth}
    \centering
    \includegraphics[scale=0.45]{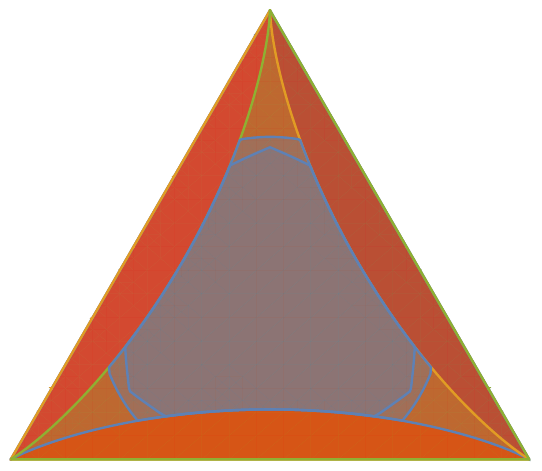}
   \caption{}
    \label{fig:c}
  \end{subfigure}
  \hfill
  \begin{subfigure}[b]{0.23\textwidth}
    \centering
    \includegraphics[scale=0.45]{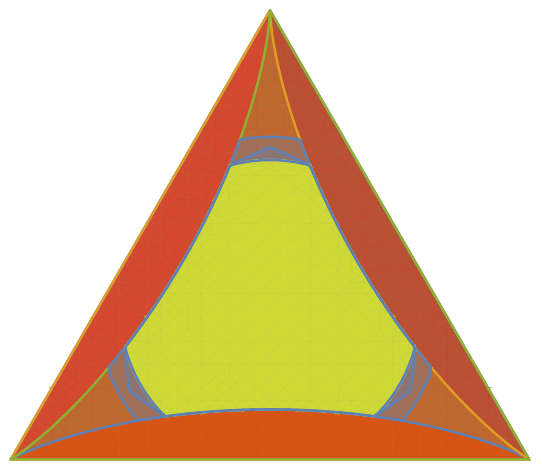}
    \caption{}
    \label{fig:d}
  \end{subfigure}
\caption{Geometric representation of the probability distribution $q=(q_0,q_1,q_2)$ and the construction of the feasible region. (a) Equilateral triangle associated with the probability distribution $q=(q_0,q_1,q_2)$, with its unistochastic subset. All plots are  in the $\delta_2\delta_1$ plane. Each panel  corresponds to a constraint applied in succession. See Remark \ref{Remark-FeasibleRegion} for details.  }
\label{FIG-Triangle1}
\end{figure}

\section{Derivation of the eigenvalues of  $G$}\label{Appendix-Eigenvalues-G}
To proceed with the  calculation of eigenvalues of $G$, we first note that the Gram matrix $G$ can be represented in a block form as
\begin{equation}\label{}
G=\begin{pmatrix}
  \Id_{d-1} & N_{d-1} & \cdots & N_{d-1} \\
 N_{d-1} & \Id_{d-1} & \cdots & N_{d-1} \\
  \vdots & \vdots & \ddots & \vdots \\
  N_{d-1} & N_{d-1} & \cdots & \Id_{d-1}
\end{pmatrix},
\end{equation}
where $\Id_{d-1}$ is the unit matrix of dimension  $(d-1)$ and the $(d-1)\times(d-1)$ matrix $N_{d-1}$, defined by Eq. \eqref{N-k-d-k-2}, forms the blocks of the  Gram matrix $G$. Note here that, to maintain clarity, we explicitly denote the size of each submatrix with a subscript. This simple block form of $G$ allows us to provide the  direct product representation
\begin{equation}\label{Gdirectproduct}
G=\Id_{L}\otimes \Id_{d-1}+S_{L}\otimes N_{d-1},
\end{equation}
where $S_{L}$ is an $L\times L$ matrix with all off-diagonal entries equal to unity, and the diagonal entries are equal to zero. Both $S_L$ and $N_{d-1}$ can be diagonalized unitarily  as
\begin{eqnarray}
S_{L}=Q_L D_L Q_L^\dagger, \qquad N_{d-1}=U_{d-1} \Lambda_{d-1} U_{d-1}^\dagger.
\end{eqnarray}
Here
\begin{eqnarray}
D_L&=&\diag\{L-1,-1,\ldots,-1\}, \\
\Lambda_{d-1}&=&\diag\{\lambda_1,-\lambda_1, \ldots,\lambda_{\lfloor \frac{d-1}{2} \rfloor},-\lambda_{\lfloor \frac{d-1}{2} \rfloor}, \lambda_{\frac{d}{2}}\},
\end{eqnarray}
where $\lambda_{i}=|N_{i,d-i}|$ and $\lfloor \frac{d-1}{2} \rfloor$ denotes the integral part of $\frac{d-1}{2}$, i.e., $\lfloor\frac{d-1}{2}\rfloor=\frac{d-1}{2}$ if $d$ is odd, and  $\lfloor \frac{d-1}{2} \rfloor=\frac{d}{2}-1$ if $d$ is even. Note that all $d-1$ eigenvalues of $N_{d-1}$ occur in positive-negative pairs, except for the last one, which appears alone. Obviously, this lone eigenvalue, i.e.,  $\lambda_{d/2}$,  exists only when $d$ is even.

Having these diagonal forms, one can write $G$ as
\begin{equation}\label{Gdirectproduct}
G=[Q_L\otimes U_{d-1}]G_\textrm{diag}[Q_L^\dagger \otimes U_{d-1}^\dagger],
\end{equation}
where $G_\textrm{diag}$ is the diagonal representation of the Gram matrix $G$,
\begin{equation}\label{Gdirectproduct}
G_\textrm{diag}=\Id_{L}\otimes \Id_{d-1}+D_{L}\otimes \Lambda_{d-1}.
\end{equation}
To classify the eigenvalues of $G$, we discuss even and odd cases, separately.

\textit{Eigenvalues of $G$ when $d$ is odd.} In this case,  for $i\in\{1,\ldots,(d-1)/2\}$,  the nondegenerate eigenvalues are
\begin{eqnarray}
1\pm(L-1)|N_{i,d-i}|,
\end{eqnarray}
while the degenerate ones, each with multiplicity $L-1$, are
\begin{eqnarray}
1\pm|N_{i,d-i}|.
\end{eqnarray}

\textit{Eigenvalues of $G$ when $d$ is even.} In this case, the eigenvalues are the same as above, here for $i\in\{1,\ldots,d/2-1\}$, along with  the  two new eigenvalues
 $1+(L-1)|N_{d/2,d/2}|$ and  $1-|N_{d/2,d/2}|$, which have multiplicities $1$ and $L-1$, respectively.

For both cases of even and odd  $d$, and $L\ge 2$, we have
\begin{eqnarray}
\lambda_{\max}(G)&=&1+(L-1)\max_{i}|N_{i,d-i}|,\\
\lambda_{\min}(G)&=&1-(L-1)\max_{i}|N_{i,d-i}|,
\end{eqnarray}
for the largest and smallest eigenvalue of the Gram matrix $G$, respectively. The positive definiteness of the Gram matrix requires that $\max_{i}|N_{i,d-i}|<1/(L-1)$.


\end{document}